\documentclass[a4paper,10pt]{article}
\usepackage[bbgreekl]{mathbbol}
\usepackage[small]{titlesec}
\usepackage[utf8]{inputenc}
\usepackage[T1]{fontenc}  
\usepackage[]{mdframed}
\usepackage{mathrsfs}
\usepackage[backend=biber,style=alphabetic,natbib=false,giveninits=true,doi=true,isbn=false,url=false,date=year,maxbibnames=99,sorting=nty,defernumbers=true]{biblatex}
\DeclareNameAlias{default}{family-given}

\DeclareFieldFormat[article]{title}{#1}
\renewbibmacro{in:}{}
\DeclareFieldFormat[article]{pages}{#1}
\DeclareFieldFormat{journal}{#1}
\renewcommand\bf\bfseries

\renewbibmacro*{volume+number+eid}{%
	\printfield{volume}%
	\setunit*{\addnbspace}
	\printfield{number}%
	\setunit{\addcomma\space}%
	\printfield{eid}}
\DeclareFieldFormat[article]{number}{\mkbibparens{#1}}
\DeclareFieldFormat[article]{volume}{\textbf{#1}}
\DeclareFieldFormat{year}{\mkbibparens{#1}}
\DeclareBibliographyDriver{article}{%
	\printnames{author}:%
	\newunit\newblock
	\printfield{title}%
	\newunit\newblock
	\printfield{journaltitle}
	\newunit
	\iffieldundef{number}{\printfield{volume}}{\printfield{volume}\addspace\printfield{number}}
	\addcomma\addspace\printfield{pages}\addspace
	\printfield{year}
}
\AtEveryBibitem{\clearfield{month}}
\AtEveryBibitem{\clearfield{day}}
\usepackage[english]{babel}
\usepackage[T1]{fontenc}
\usepackage{csquotes}
\usepackage{bbm}
\usepackage[leqno]{amsmath}
\usepackage{amsfonts,amsthm,amsbsy,amssymb,dsfont,stmaryrd}
\usepackage[dvipsnames]{xcolor}
\usepackage{braket}

\usepackage{caption}
\usepackage{subcaption}

\usepackage{enumitem}

\makeatletter
\newcommand{\leqnomode}{\tagsleft@true\let\veqno\@@leqno}
\newcommand{\reqnomode}{\tagsleft@false\let\veqno\@@eqno}
\makeatother

\usepackage{mathtools}
\usepackage[makeroom]{cancel}

\numberwithin{equation}{section}

\newcommand\myshade{85}
\colorlet{mylinkcolor}{violet}
\colorlet{mycitecolor}{YellowOrange}
\colorlet{myurlcolor}{Aquamarine}

\usepackage[left=2.5cm,right=2.5cm,top=2.5cm,bottom=2.5cm]{geometry}
\usepackage[unicode=true,pdfusetitle,
bookmarks=true,bookmarksnumbered=false,bookmarksopen=false,
breaklinks=false,pdfborder={0 0 1},backref=false,
linkcolor  = mylinkcolor!\myshade!black,
citecolor  = mycitecolor!\myshade!black,
urlcolor   = myurlcolor!\myshade!black,
colorlinks = true,
]
{hyperref}
\usepackage[nameinlink]{cleveref}

\usepackage{tikz}
\usetikzlibrary{hobby}
\usetikzlibrary{arrows}
\usetikzlibrary{fit,backgrounds}
\usetikzlibrary{decorations.pathreplacing,decorations.markings}
\tikzset{arrow data/.style 2 args={%
      decoration={%
         markings,
         mark=at position #1 with \arrow{#2}},
         postaction=decorate}
      }%

\usepackage{tikz-cd}
\usepackage{ifthen}

\usepackage{graphicx}
\usepackage{caption}
\usepackage{slashed}
\definecolor{ct_black}{HTML}{000000}
\definecolor{ct_orange}{HTML}{ED872D}
\definecolor{ct_purple}{HTML}{7A68A6}
\definecolor{ct_blue}{HTML}{348ABD}
\definecolor{ct_turquoise}{HTML}{188487}
\definecolor{ct_red}{HTML}{E32636}
\definecolor{ct_pink}{HTML}{CF4457}
\definecolor{ct_green}{HTML}{467821}

\definecolor{ct2_green}{HTML}{9FF781}
\definecolor{ct2_green_dark}{HTML}{088A08}

\theoremstyle{plain}
\newtheorem{thm}{\protect\theoremname}[section]
\theoremstyle{plain}
\newtheorem{lem}[thm]{\protect\lemmaname}
\theoremstyle{plain}
\newtheorem{cor}[thm]{\protect\corollaryname}
\theoremstyle{plain}
\newtheorem{prop}[thm]{\protect\propositionname}
\theoremstyle{plain}

\theoremstyle{remark}
\newtheorem{rem}[thm]{\protect\remarkname}

\theoremstyle{definition}
\newtheorem{defn}[thm]{\protect\definitionname}
\theoremstyle{plain}
\newtheorem{example}[thm]{\protect\examplename}
\providecommand{\assumptionname}{Assumption}
\providecommand{\claimname}{Claim}
\providecommand{\corollaryname}{Corollary}
\providecommand{\definitionname}{Definition}
\providecommand{\lemmaname}{Lemma}
\providecommand{\propositionname}{Proposition}
\providecommand{\remarkname}{Remark}
\providecommand{\theoremname}{Theorem}
\providecommand{\examplename}{Example}

\crefname{section}{Section}{Sections}
\crefname{example}{Example}{Examples}
\crefname{appendix}{Appendix}{Appendices}
\crefname{figure}{Figure}{Figures}
\crefname{assumption}{Assumption}{Assumptions}
\crefname{thm}{Theorem}{Theorems}
\crefname{lem}{Lemma}{Lemmas}
\crefformat{equation}{(#2#1#3)}
\crefname{table}{Table}{Tables}

\crefrangelabelformat{equation}{(#3#1#4--#5#2#6)}

\crefmultiformat{equation}{(#2#1#3}{, #2#1#3)}{#2#1#3}{#2#1#3}
\Crefmultiformat{equation}{(#2#1#3}{, #2#1#3)}{#2#1#3}{#2#1#3}

\newtheorem*{lem*}{\protect\lemmaname}

\newcommand{\ii}{\operatorname{i}}

\newcommand{\ZZ}{\mathbb{Z}}

\newcommand{\bS}{\mathbb{S}}

\newcommand{\NN}{\mathbb{N}}
\newcommand{\RR}{\mathbb{R}}
\newcommand{\QQ}{\mathbb{Q}}
\newcommand{\CC}{\mathbb{C}}

\newcommand{\calB}{\mathcal{B}}

\newcommand{\calF}{\mathcal{F}}
\newcommand{\calN}{\mathcal{N}}
\newcommand{\calG}{\mathcal{G}}

\newcommand{\calD}{\mathcal{D}}

\newcommand{\calU}{\mathcal{U}}
\newcommand{\calSU}{\mathcal{SU}}
\newcommand{\calSG}{\mathcal{SG}}
\newcommand{\calH}{\mathcal{H}}

\newcommand{\calK}{\mathcal{K}}
\newcommand{\calL}{\mathcal{L}}

\newcommand{\calJ}{\mathcal{J}}

\newcommand{\calS}{\mathcal{S}}

\newcommand{\bbLambda}{\mathbb{\Lambda}}
\newcommand{\bbDelta}{\mathbb{\Delta}}

\newcommand{\ti}[1]{\widetilde{#1}} 
\newcommand{\LamNT}{{\Lambda\mathrm{nt}}}

\newcommand\norm[1]{\left\lVert#1\right\rVert}

\newcommand{\dif}{\operatorname{d}\!} 
\newcommand{\tr}{\operatorname{tr}}

\newcommand{\szpan}{\operatorname{span}}

\newcommand{\ve}{\varepsilon}
\newcommand{\vf}{\varphi}
\newcommand{\Id}{\mathds{1}}

\newcommand{\dist}{\mathrm{dist}}

\newcommand{\sgn}{\operatorname{sgn}}
\newcommand{\findex}{\operatorname{ind}}

\newcommand{\im}{\operatorname{im}}

\newcommand{\diam}{\operatorname{diam}}
\newcommand{\sigmaess}{\sigma_{\mathrm{ess}}}

\usepackage{environ}

\NewEnviron{malign}{%
	\begin{align}\begin{split}
			\BODY
	\end{split}\end{align}
}

\newcommand{\eq}[1]{\begin{align*}#1\end{align*}}
\newcommand{\eql}[1]{\begin{align}#1\end{align}}

\newcommand{\polar}{\operatorname{pol}}

\newcommand{\calT}{\mathcal{T}}

\title{Topological Classification of Insulators:\\ II. Quasi-Two-Dimensional Locality}
\author{\href{mailto:jc1220@math.princeton.edu}{Jui-Hui Chung}\\
	{\footnotesize Program in Applied and Computational Mathematics, Princeton University }\\
	\href{mailto:jacobshapiro@princeton.edu}{Jacob Shapiro}\\
	{\footnotesize Department of Mathematics, Princeton University}
}

\begin{document}
	\reqnomode
	
	\maketitle
		\begin{abstract}
		We provide an alternative characterization of two-dimensional locality (necessary e.g. to define the Hall conductivity of a Fermi projection) using the spectral projections of the Laughlin flux operator. Using this abstract characterization, we define generalizations of this locality, which we term quasi-2D. We go on to calculate the path-connected components of spaces of unitaries or orthogonal projections which are quasi-2D-local and find a starkly different behavior compared with the actual 2D column of the Kitaev table, exhibiting e.g., in the unitary chiral case, infinitely many $\ZZ$-valued indices.
	\end{abstract}
	\tableofcontents
	\section{Introduction}
    Topological insulators \cite{Hasan_Kane_2010} are exotic solids which are insulating in their bulk but which are robust conductors along their boundary. Many of their special properties, e.g., robust quantization of experimental quantities, are explained by the fact the space of quantum mechanical Hamiltonians associated with such classes of materials is \emph{not} path-connected. Hence the topological non-triviality is not meant in real physical space but rather in the abstract space of systems. To mathematically unpack this statement one must define the ambient topology on spaces of operators on Hilbert space which correspond to (quantum mechanical) Hamiltonians associated with solids and calculate its path-connected components ($\pi_0$ henceforth). The problem of calculating $\pi_0$ of various spaces of operators on Hilbert space is of course classical, e.g., Kuiper's theorem on the path-connectedness of invertibles, or the Atiyah-J\"anich theorem whose consequence is that $\pi_0(\calF)\cong\ZZ$. In the condensed matter physics setting of topological insulators, however, one needs further constraints on operators and so, after appropriate translation from quantum mechanics into functional analysis, we are left with a classification problem of custom spaces of operators.
    
    We now describe the problem more precisely. Define a space of Hamiltonians (bounded self-adjoint linear operators on the Hilbert space $\calH = \ell^2(\ZZ^d)\otimes\CC^N$) with three additional constraints: (1) locality, (2) the insulator condition and (3) possible symmetry. Define a topology on this space. Calculate $\pi_0$ and show that the bijection from that set into a set of numbers (in case the space ends up being not path-connected) is the lift of a physically measurable quantity $\calN$, such as electric conductivity. This special measurable quantity $\calN$ is referred to as \emph{topological invariant} or \emph{index}. We refer the reader to \cite{ChungShapiro2023} for more context on this problem as well as existing literature on the topic (we avoid repeating this survey here).

    There are various levels of generality to the insulator \emph{and} locality condition, each of which corresponds to some interesting aspect of the problem. E.g. one may define insulators via the theory of Anderson localization. Doing so is conjectured not to affect the classification. At present we are rather concerned with different notions of \emph{locality}. In \cite{ChungShapiro2023} we studied the $d=1$ case using the simplest insulator condition, which stipulates that $H$ is an insulator iff $0\notin\sigma(H)$; however, we generalized one-dimensional locality in a convenient way (although we refer the reader to \cite[Section 5.6]{ChungShapiro2023} for a discussion of classification without generalizing locality). 
    
    Locality in quantum mechanics stems from the basic physical principle that interaction strength diminishes with distance. Say a system follows a time evolution generated by the Hamiltonian $H$. Then the quantum mechanical amplitude to transition from state $\vf$ to state $\psi$ is associated with the expression $\braket{\psi,H\vf}$. If these vectors are two wave-packets concentrated in far away points in space then this number should decay with that distance. For this reason many popular models in physics are in fact merely \emph{nearest-neighbor} and for that reason it is common in physics to work with operators $A$ that are $R$-hopping-local  (for some fixed $R\in\NN$), i.e., \eq{A_{xy}\equiv\braket{\delta_x,A\delta_y}=0\qquad(x,y\in\ZZ^d:\norm{x-y}>R)} where $\Set{\delta_x}_{x\in\ZZ^d}$ is the standard (position) basis for $\ell^2(\ZZ^d)$. More common in mathematical physics (and defining a norm-closed set) is \emph{exponential locality}, given by the condition that $A\in\calB(\calH)$ is exponentially local iff there exist $C,\mu\in(0,\infty)$ such that \eql{\label{eq:exp-locality} \norm{A_{xy}}\leq C \exp\left(-\mu\norm{x-y}\right)\qquad(x,y\in\ZZ^d)\,.} One could instead take different, logarithmic metrics to get various rates of polynomial decay. We want to generalize this notion in a more drastic way, guided by the principle that any generalization should keep the topological invariant $\calN$ well-defined. In one dimension, it was possible to relax exponential locality while keeping $\calN$ as follows. On $\calH:=\ell^2(\ZZ)\otimes\CC^N$, with $X$ the position operator, let \eql{\label{eq:Lambda in 1D}\Lambda := \chi_\NN(X)} be the self-adjoint projection onto the right-hand side of space. Then we define an operator $A\in\calB(\calH)$ to be $\Lambda$-local iff \eql{\label{eq:Lambda-locality}[A,\Lambda]\in\calK\,.}

    We observe that \cref{eq:exp-locality} implies \cref{eq:Lambda-locality} (see e.g. \cite[Lemma 2 (b)]{Graf_Shapiro_2018_1D_Chiral_BEC}) whereas the converse is false. Indeed, operators which are $\Lambda$-local allow for hopping of arbitrary distance just so long as they are limited to strictly the left or the right-hand side of the system; it is merely the hopping from left to right and vice versa which should be compact. 
    
    Adopting the philosophy (customary in condensed matter physics) that a particular model is unimportant and rather we focus on studying the mathematically simplest question which still yields non-trivial physical results, one should not object to replacing \cref{eq:Lambda-locality} with \cref{eq:exp-locality}: the former yields $\pi_0$ that agrees with the one-dimensional column of the Kitaev table \cite{Kitaev2009}, as was shown in \cite{ChungShapiro2023}. A mathematical physicist might want then to carry out a separate study for exponential locality (as in \cite[Section 5.6]{ChungShapiro2023}) or, ideally, show that exponential locality is a deformation retract of the weaker generalized locality.

    What is the analog of $\Lambda$-locality in two dimensions? Unsurprisingly, as we shift dimensions the roles of unitaries and projections swap; one is guided by the Laughlin flux insertion formulation of the Hall conductivity, which in its index form first appeared in \cite{Bellissard_1994JMP....35.5373B} (we refer the reader to that paper for context on where the operator $L$ below comes from). On $\calH := \ell^2(\ZZ^2)\otimes\CC^N$, define the so-called Laughlin flux insertion unitary operator as \eql{\label{eq:Laughlin flux operator} L := \exp\left(\ii\arg\left(X_1+\ii X_2\right)\right)} where $X_1,X_2$ are the two position operators on $\calH$. We term an operator $A\in\calB$ as $L$-local iff \eql{\label{eq:L-locality} [A,L]\in\calK\,. } It is true that \cref{eq:exp-locality} implies \cref{eq:L-locality} (see e.g. \cite[Lemma A.1]{BSS23}) and again the converse is false. If we think about the operator $L$ as the polar part of the position operator considered as a complex-valued number on the plane, $L$-locality roughly corresponds to hopping "compactly" between rays but arbitrarily far radially within rays, which is intuitively why this notion is weaker than exponential locality. Moreover, it is well-known that \cref{eq:L-locality} suffices to get well-defined topological indices, in particular, the two-dimensional Chern number of the integer quantum Hall effect (this result first appeared in \cite{Bellissard_1994JMP....35.5373B}).
	
	\begin{table}[b]
		\begin{center}
			\begin{tabular}{c|ccc|cccccccc}
				\multicolumn{4}{c|}{Symmetry } & \multicolumn{8}{c}{dimension} \\
				\multicolumn{1}{c}{AZ} &$\hspace{1.5mm}\Theta\hspace{1.5mm} $ &
				$\hspace{1.5mm} \Xi\hspace{1.5mm} $ &
				$\hspace{1.5mm} \Pi\hspace{1.5mm} $ &
				$1$   &  $2$ &  $3$ &  $4$ &  $5$ & $6$ & $7$& $8$ \\
				\hline
				A & $0$ & $0$ & $0$  &$0$& $\mathbb{Z}$ &$0$& $\mathbb{Z}$ &$0$& $\mathbb{Z}$ &$0$& $\mathbb{Z}$\\
				AIII & $0$ & $0$ & $1$ & $\mathbb{Z}$ &$0$& $\mathbb{Z}$ &$0$& $\mathbb{Z}$ &$0$& $\mathbb{Z}$& $0$\\
				\hline
				AI & $1$ & $0$ & $0$  &$0$&$0$&$0$&$\mathbb{Z}$&$0$&$\mathbb{Z}_2$&$\mathbb{Z}_2$& $\mathbb{Z}$ \\
				BDI & $1$ &$1$ &$1$ & $\mathbb{Z}$ &$0$&$0$&$0$&$\mathbb{Z}$&$0$&$\mathbb{Z}_2$& $\mathbb{Z}_2$\\
				D & $0$ &$1$ &$0$ & $\mathbb{Z}_2$& $\mathbb{Z}$ &$0$&$0$&$0$&$\mathbb{Z}$&$0$&$\mathbb{Z}_2$\\
				DIII&$-1$ &$1$ &$1$ &$\mathbb{Z}_2$& $\mathbb{Z}_2$& $\mathbb{Z}$ &$0$&$0$&$0$&$\mathbb{Z}$&$0$\\
				AII & $-1$ & $0$ & $0$ &$0$&$\mathbb{Z}_2$& $\mathbb{Z}_2$& $\mathbb{Z}$ &$0$&$0$& $0$&$\mathbb{Z}$\\
				CII & $-1$ &$-1$ & $1$&$\mathbb{Z}$ & $0$&$\mathbb{Z}_2$& $\mathbb{Z}_2$& $\mathbb{Z}$ &$0$&$0$&$0$ \\
				C & $0$ & $-1$& $0$ & $0$ &$\mathbb{Z}$ &$0$&$\mathbb{Z}_2$& $\mathbb{Z}_2$& $\mathbb{Z}$ &$0$& $0$\\
				CI & $1$ & $-1$ & $1$& $0$ & $0$&$\mathbb{Z}$&$0$&$\mathbb{Z}_2$& $\mathbb{Z}_2$& $\mathbb{Z}$& $0$ \\
			\end{tabular}
		\end{center}
		\caption{The Kitaev periodic table. The entries stand for the respective K-theory groups in a given dimension and symmetry class.}
		\label{table:Kitaev}
	\end{table}

    Hamiltonians break into two main classes for purposes of classification, depending on the presence or absence of a so-called "chiral" symmetry: unitaries and self-adjoint unitaries (this is really patterned after K-theory). Indeed, flat Hamiltonians (such that $H=\sgn(H)$) are self-adjoint unitary operators, and in the chiral case, each flat Hamiltonian is in one-to-one correspondence with a unitary operator, directly via the definition of chiral symmetry. The space of flat Hamiltonians is a deformation retract of non-flat ones. Hence, given the first two rows of the $d=2$ column of \cref{table:Kitaev}, we are led to ask the

    \emph{Question}: Is the space of $L$-local unitaries path-connected, whereas $\pi_0$ of the space of $L$-local orthogonal projections (which are just self-adjoint unitaries via $P\equiv\frac12(\Id-V)$) bijective with $\ZZ$, under the bijection \eql{ P \mapsto \findex\left(P LP +P^\perp\right)\in\ZZ\,? }
    Not necessarily, since, at least in \cite{ChungShapiro2023} for $d=1$ we had to require certain "$\Lambda$-non-triviality" conditions on self-adjoint unitaries in order to match with the Kitaev table. These essentially amounted to the fact that the system is honestly infinite-dimensional as a bulk system and not only as an edge system.  We shall make similar assumptions in the sequel, acknowledging that the naive locality and insulator conditions we have stipulated so far are probably not sufficient to yield the expected $\pi_0$ according to \cref{table:Kitaev}.

    In this report we do not directly address $\pi_0$ of $L$-local systems; this is the subject of ongoing work. Instead, we give an abstract characterization of $L$-locality and using that, explore other forms of locality which also could have been plausible in two dimensions and show that they lead to starkly different results than those predicted by the first two rows of the two-dimensional column of \cref{table:Kitaev}. We present this first and foremost as a curious study of spaces of operators, which may or may not be relevant to physics--we don't know.

    This manuscript is organized as follows. In \cref{sec:Abstract L locality} we re-cast $L$-locality using the spectral projections of $L$, in an abstract way that does away with the geometry of $\ZZ^2$. In \cref{sec:other forms of locality} we use this abstract characterization as a springboard for generalizations into other forms of plausible locality. In \cref{sec:classification of quasi2D unitaries,sec:classification of quasi 2D SAUs} we calculate $\pi_0$ of unitaries or self-adjoint unitaries which obey these generalized forms of locality. It is there that we exhibit an infinite number of $\ZZ$-valued indices. Finally, in \cref{sec:physical consequences} we discuss some physical consequences of this scheme, which are especially compelling for star-shaped wire systems. \cref{sec:operator integrals} contains some standard results about integrals of operator-valued functions which we include here for convenience.
    \begin{rem}
        In principle one could also embark on a classification of these unitaries and self-adjoint unitaries which obey various real symmetries (corresponding to the Altland-Zirnbauer classes), as we have done in the one-dimensional case in \cite{ChungShapiro2023}. We leave this question to a later occasion since it is somewhat perpendicular to our present line of questioning.
    \end{rem}

    \section{Abstract $L$-locality}\label{sec:Abstract L locality}
    Towards studying $\pi_0$ of $L$-local operators, it is sensible to expect that there is utility in rephrasing $L$-locality \cref{eq:L-locality} in other, equivalent ways. We would like to work towards a "representation independent" characterization of $L$-locality. For example, let $\ti\calH $ be any other separable Hilbert space and let $R$ be a unitary operator on $\ti \calH$ s.t. $\sigma(R)=\bS^1$. We say that $T\in\calB(\ti \calH)$ is $R$-local iff $[T,R]\in\calK(\ti\calH)$. Then we have the following
    \begin{prop}
    The space of $L$-local operators is isometrically isomorphic to the space of $R$-local operators.
    \end{prop}
    \begin{proof}
    We will use the Weyl-von Neumann-Berg theorem \cite[Theorem 39.8]{conway2000course} which says that if two normal operators on two separable Hilbert spaces have the same essential spectrum, then they are unitarily equivalent modulo a compact operator. Here, since $L,R$ are normal, their essential spectra equal $\bS^1$ and in particular are equal. Therefore, there exists a unitary operator $U:\calH\to\ti\calH$ and a compact operator $K\in\calB(H)$ such that
    \eq{
        L = U^\ast R U + K\,.
    }
    We argue that the map $T\mapsto U^\ast T U$ sends $R$-local operators to $L$-local operators and is an isometric isomorphism, whose inverse is $A\mapsto UAU^\ast$. Indeed, let $T$ be $R$ local. Then
    \eq{
        [U^\ast TU,L] = [U^\ast TU,U^\ast R U + K] = U^\ast[T,R]U + U^\ast TUK + KU^\ast T U\in\calK
    }
    and hence $U^\ast TU$ is $L$-local.
    \end{proof}

    \begin{cor}[Existence of arbitrary index $R$-local projections] On $\ell^2(\ZZ)$, let $R$ be the bilateral right shift operator and $\Lambda$ as in \cref{eq:Lambda in 1D}. In \cite{ChungShapiro2023} we saw that $\bbLambda R\equiv\Lambda R \Lambda+\Lambda^\perp$ is a Fredholm operator and its index is $-1$. A natural question is: is there an $R$-local projection $P$ (so $\mathbb{P} R$ is automatically Fredholm) for which $\mathbb{P} R$ can take any arbitrary integer? The answer is yes by the above proposition.
    \end{cor}
    \begin{proof}
        Consider $R^k$ for $k\in\ZZ$. Then $\Lambda$ is $R^k$-local and $\findex \bbLambda R^k=-k$. By the Weyl-von Neumann-Berg theorem, there exists a unitary operator $U:\ell^2(\ZZ)\to\ell^2(\ZZ)$ such that $R^k=U^\ast RU+K$. Since  $\Lambda$ is $R^k$-local, it follows that the projection $P:=U\Lambda U^\ast$ is $R$-local. In particular \eq{\mathbb{P} R = U\Lambda U^\ast R U\Lambda U^\ast + U\Lambda^\perp U^\ast = U \bbLambda (R^k-K) U^\ast} has index $-k$.
    \end{proof}
    
    \begin{rem}[Importance of spectral type and the spectrum of the unitary operator]\label{rem:spectral type of L}
        The preceding example illustrates that the spectral type of the chosen unitary operator with which to define locality does not matter. Indeed, the operator $L$ has a pure-point spectrum, while the bilateral right shift $R$ has absolutely continuous spectrum. However, it is important that the spectrum covers the whole circle so that the Weyl-von Neumann-Berg theorem applies. Indeed, if on $\calH$ we choose for locality a unitary operator $V$ whose spectrum is not the whole circle, then the $V$-local projections (those projections $P$ such that $[P,V]\in\calK$) necessarily have vanishing index, i.e., $\findex \mathbb{P}V=0$ (see \cite[Theorem 3.1]{brown1973unitary}). This space of $V$-local projections might still be not path-connected, but we would exclude such $V$ for the following reasons. With such a choice for $V$, one could not use this formalism to understand the integer quantum Hall effect, since $\findex \mathbb{P}V$ is necessarily the formula for the Hall conductance. We postpone the study of $V$-locality with $V$ not covering the circle to another occasion.
    \end{rem}

    From now on, we will let $\calH$ be an abstract separable Hilbert space and $L$ an abstract unitary operator on $\calH$ whose spectrum covers the circle $\bS^1$. We say that an operator $A\in\calB(\calH)$ is $L$-local iff
    \eql{\label{eq:abstract L-locality}
        [A,L]\in\calK\,.
    }

    Our next goal is to characterize $L$-locality in terms of the spectral projections of $L$ rather than $L$ itself. To motivate the analysis in the sequel, we present a heuristic argument. Let $J\subseteq \bS^1$ be an interval and let $\Lambda_J=\chi_J(L)$ be the spectral projection of $L$. Let $\gamma$ be a contour that surrounds $J$ with $\gamma$ intersecting $\bS^1$ only on the endpoints of $J$ (see \cref{fig:contour}). If 
    \eql{
        \Lambda_J = \frac{-1}{2\pi\ii}\oint_\gamma R(z)\dif{z} \label{eq:spectral projection as contour integral}
    }
    holds, where $R(z) \equiv (L-z\Id)^{-1}$ is the resolvent of $L$, then we have
    \eql{
        [A,\Lambda_J] &= \frac{-1}{2\pi\ii}\oint_\gamma\left[A, R(z)\right]\dif{z} \notag\\
        &=  \frac{1}{2\pi\ii}\oint_\gamma  R(z)\left[A,L\right]R(z)\dif{z} \,.\label{eq:commutator resolvent troublesome}
    }
    Since $[A,L]\in\calK$, the above suggests that $[A,\Lambda_J]\in\calK$ for all intervals $J\subseteq\bS^1$. The trouble with this argument is that \cref{eq:spectral projection as contour integral} does not hold in general, e.g., if the contour passes through an eigenvalue. Suppose for the argument's sake that \cref{eq:spectral projection as contour integral} \emph{does} hold. This argument is \emph{still} wrong: the problem is that the strong operator limit (the limit entailed in carrying out the contour integral) may not preserve compactness.

    In fact, there exists an $L$-local operator $A$ such that $[A,\Lambda_J]$ fails to be compact for some interval $J$:
    \begin{example}\label{ex:counterexample}
    On $\ell^2(\ZZ^2)$, choose $L$ as the Laughlin flux as in \cref{eq:Laughlin flux operator} and $R_1$ be the right shift operator, i.e., \eq{R_1\delta_{(x_1,x_2)} = \delta_{(x_1+1,x_2)}\qquad(x\in\ZZ^2)\,.} Since $R_1$ is a finite-hopping operator, it follows from \cite[Lemma A.1]{BSS23} that $R_1$ is $L$-local. Let \eq{I=[\pi/2,3\pi/4]\cap\QQ\subseteq \bS^1\cap\QQ\cong [0,2\pi)\cap\QQ} and let $\Lambda_I=\chi_I(L)$ be the spectral projection of $L$ onto the interval $I$. Let $Q:=\chi_{\{\pi/2\}}(L)$. Then we have
    \eq{
        R_1^\ast (\Lambda_I^\perp R_1\Lambda_I)Q = R_1^\ast (\Lambda_I^\perp R_1Q) = R_1^\ast R_1 Q = Q
    }
    where in the second equality we observe that $\im R_1Q\subseteq \im\Lambda_I^\perp$. This suggests that $[R_1,\Lambda_I]$ cannot be compact. If it is, then $\Lambda_I^\perp[R_1,\Lambda_I]=\Lambda_I^\perp R_1\Lambda_I$ will be compact and so is $Q$. However, $Q$ is a projection onto an infinite-dimensional subspace which cannot be compact.
    \end{example}
    
    Given these precautions, we now proceed with the actual analysis. Let $I,J\subseteq \bS^1\subseteq \CC$ be two intervals with $\dist(I,J)>0$. Let $\gamma$ be a contour around $J$ where $I$ lies in the exterior of $\gamma$ (see \cref{fig:contour}). In particular, both $\dist(I,\gamma),\dist(J,\gamma)>0$ are non-vanishing. We define the resolvent of $L$ \emph{on an interval $I$} as
    \eql{
        R_I(z) = \chi_I(L) (L\chi_I(L)-z\Id)^{-1}\,.
    }
    Even though the resolvent $R(z)$ of $L$ is defined only on $\CC\setminus \bS^1$, the resolvent $R_I(z)$ on an interval $I$ is defined on $\CC\setminus (\overline{I} \cup\{0\})$ since $\sigma(L\chi_I(L))\subseteq \overline{I}\cup\{0\}$. Here \eq{\frac{-1}{2\pi\ii}\oint_\gamma R_I(z)\dif{z}=\Lambda_I} and $\norm{R_I(z)}$ is uniformly bounded as $z$ varies on the counter $\gamma$. Indeed, $R_I(z)$ is integrable:
    \eq{
        \int_0^1 \norm{R_I(z)} |\gamma'(t)| \dif{t} \leq \frac{\int_0^1|\gamma'(t)|\dif t}{\dist(I,\gamma)}<\infty
    }
    using the fact that 
    \eq{
        \norm{R_I(z)}\leq \sup_{\lambda\in I} \frac{1}{\lambda\chi_I(\lambda)-z} = \frac{1}{\dist(I,z)}\leq \frac{1}{\dist(I,\gamma)}\,.
    }

    Inspired by \cref{eq:commutator resolvent troublesome}, it will turn out that the correct formula is
    \eql{
        \frac{1}{2\pi \ii}\oint_\gamma R_I(z)[A,L]R_J(z)\dif{z} =  \Lambda_IA\Lambda_J \label{eq:operator integral on laughlin commutator}
    }
    This suggests that if $A$ is $L$-local, then $\Lambda_IA\Lambda_J\in\calK$ for all intervals $I,J\subseteq\bS^1$ such that $\dist(I,J)>0$. In fact, the converse is also true and we obtain another characterization of $L$-locality:

    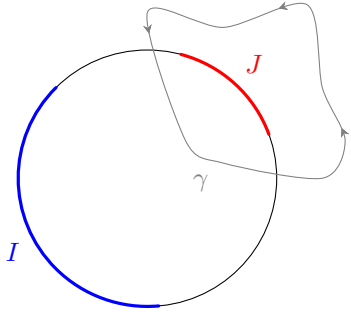
\begin{figure}[t]               
    \centering
    \begin{tikzpicture}[draw=black]
        \def\R{1.7}
        \coordinate (O) at (0,0);
        \draw (O) circle (\R);
        \draw[red,very thick,line cap=round] (20:\R) arc (20:75:\R) node[midway,above right] {$J$};
        \draw[blue,very thick,line cap=round] (135:\R) arc (135:275:\R) node[midway,below left] {$I$};
        \draw [gray,arrow data={0.25}{stealth},
               arrow data={0.5}{stealth},
               arrow data={0.75}{stealth}] plot [smooth cycle, tension=4] coordinates {(0.5,0.5) (1,0.2) (2.3,0) (2.6,0.5) (2.3,1.2) (2.1,2.3) (1,2.0) (0,2.2)};
        \node[label={[text=gray]below:$\gamma$}] (A) at (0.7,0.3) {}; 
    \end{tikzpicture}
    \caption{Contour $\gamma$ around the interval $J$ with the interval $I$ lying in the exterior of $\gamma$.}
    \label{fig:contour}
    \end{figure}

    \begin{thm}[Abstract characterization of $L$-locality]\label{thm:abstract-L-locality} On $\calH$, an operator $A$ is $L$-local as in \cref{eq:abstract L-locality} iff for any two intervals $I,J\subseteq \bS^1$ such that $\dist(I,J)>0$ (the usual notion of distance between sets on the circle) we have \eql{\label{eq:abstract-L-locality}\Lambda_I A \Lambda_J\in\calK} where $\Lambda_I=\chi_I(L)$ is the spectral projections of $L$.
    \end{thm}
    \begin{proof}        
        
        Let us assume that $[A,L]$ is compact and establish \cref{eq:abstract-L-locality} for two intervals $I,J$ with strictly positive distance between them. For the integral in \cref{eq:operator integral on laughlin commutator}, the integrability condition \cref{eq:integralbility condition} is satisfied. Indeed, we have
        \eq{
            \int_0^1 \norm{R_I(z)}\norm{R_J(z)} |\gamma'(t)|\dif t \leq \frac{\int_0^1|\gamma'(t)|\dif t}{\dist(I,\gamma)\dist(J,\gamma)}<\infty \,.
        }
        Since $[A,L]\in\calK$, it follows from \cref{prop:double operator integrals over compact operators} that the LHS of \cref{eq:operator integral on laughlin commutator} is compact.

        Let us now establish the equality within \cref{eq:operator integral on laughlin commutator}. We have
        \eq{
            \braket{\psi, \oint_\gamma R_I(z)[A,L]R_J(z)\dif{z} \vf} &\equiv \oint_\gamma \braket{\psi,R_I(z)[A,L]R_J(z)\vf}\dif{z} \\
            &= \oint_\gamma \tr \vf\otimes \psi^\ast R_I(z)[A,L]R_J(z) \dif{z} \\
            &= \oint_\gamma \tr A [L,R_J(z)\vf\otimes\psi^\ast R_I(z)]\dif{z} \tag{cyclicity}\\
            &= \tr A \oint_\gamma [L,R_J(z)\vf\otimes\psi^\ast R_I(z)] \dif{z}\,.
        }
        In the last equality, we used \cref{prop:properties of weak integrals} to interchange trace and integral. We now evaluate the contour integral. Let $E$ be the spectral measure of $L$. We have
        \eq{
            \braket{\eta,\oint_\gamma LR_J(z)\vf\otimes\psi^\ast R_I(z) \dif{z}\xi} &\equiv \oint \braket{\eta, LR_J(z)\vf\otimes\psi^\ast R_I(z)\xi} \dif{z} \\
            &= \oint_\gamma \braket{\eta,LR_J(z)\vf}\braket{\psi,R_I(z)\xi}  \dif{z} \\
            &= \oint_\gamma \int_{J}\frac{\mu}{\mu-z}\dif{E_{L,\eta,\vf}(\mu)} \int_{I}\frac{1}{\lambda-z}\dif{E_{L,\psi,\xi}(\mu)}\,.
        }
        The integrals here involve only finite measures and bounded functions. Thus Fubini's theorem applies and we can interchange the integrals freely. To that end, we first compute
        \eq{
        \int_\gamma \frac{1}{\mu-z} \frac{1}{\lambda-z}\dif z = 2\pi \ii \frac{1}{\mu-\lambda}
        }
        using Cauchy's integral formula, and we get
        \eq{
        2\pi \ii \int_{I} \left(\int_{J} \frac{\mu}{\mu-\lambda} \dif{E_{L,\eta,\vf}(\mu)}\right) \dif{E_{L,\psi,\xi}(\lambda)} \,.
        }
        In a similar fashion, we have
        \eq{
            \braket{\eta,\oint R_J(z)\vf\otimes\psi^\ast R_I(z)L \dif{z}\xi} = 2\pi \ii \int_{I}\left(\int_{J} \frac{\lambda}{\mu-\lambda}  \dif{E_{L,\eta,\vf}(\mu)}\right) \dif{E_{L,\psi,\xi}(\lambda)} \,.
        }
        Therefore
        \eq{
            \braket{\eta,\oint [L,R_J(z)\vf\otimes\psi^\ast R_I(z)] \dif{z}\xi} &= 2\pi \ii \int_{I}\left(\int_{J}  \dif{E_{L,\eta,\vf}(\mu)}\right) \dif{E_{L,\psi,\xi}(\lambda)} \\
            &= 2\pi \ii \braket{\eta,E(J) \vf}\braket{\psi,E(I)\xi} \\
            &= 2\pi \ii \braket{\eta, \Lambda_J \vf\otimes \psi^\ast \Lambda_I \xi}\,.
        }
        Then
        \eq{
            \tr A \oint [L,R_J(z)\vf\otimes\psi^\ast R_I(z)] \dif{z} = 2\pi \ii \tr A \Lambda_J \vf\otimes \psi^\ast \Lambda_I = 2\pi \ii \braket{\psi, \Lambda_IA\Lambda_J\vf}
        }
        and since $\vf,\psi$ were arbitrary, this concludes one direction of the proof.

        We now prove the converse statement, namely, given an operator $A$, we want to establish that $[A,L]\in\calK$ if all operators $\Lambda_I A \Lambda_J$ are compact, with $I,J$ having strictly positive distance. The proof consists of three steps: 
        \begin{itemize}
            \item \emph{Step 1}: We will first show that $\Lambda_I[A,L]\Lambda_J$ is compact for any two intervals $I,J\subseteq\bS^1$ with $\dist(I,J)>0$.
            \item \emph{Step 2}: In the second step, we will show that $\Lambda_I[A,L]\Lambda_J$ is compact even when the closures of $I,J$ intersect at a point. In this case, even though $\Lambda_IA\Lambda_J$ fails to be compact, with the help of the commutator with $L$, the expression $\Lambda_I[A,L]\Lambda_J$ recovers compactness.
            \item \emph{Step 3}: The third step is to show that $[A,L]$ is indeed compact.
        \end{itemize}  
        
        \paragraph{Step 1}
        First of all, we have 
        \eq{
        \Lambda_I[A,L]\Lambda_J = [\Lambda_IA\Lambda_J,L]\in\calK
        } 
        for any two intervals $I,J\subseteq\bS^1$ such that $\dist(I,J)>0$ by assumption. 
        
        \paragraph{Step 2} 
        Let $I,J\subseteq \bS^1$ be two intervals whose closure intersects at a point $z_0\in\bS^1$. We shall establish that $\Lambda_I[A,L]\Lambda_J\in\calK$. To that end, consider the dyadic partitions of the interval $I$, where we divide $I$ into $2, 4, 8,\dots$ sub-intervals of equal length. Let $I_n$ be the subinterval at the $n$-th partition whose closure contains the point $z_0$. We construct the sequence $J_n$ of subintervals out of $J$ similarly. We observe that 
        \eq{
        \Lambda_I[A,L]\Lambda_J - \Lambda_{I_n}[A,L]\Lambda_{J_n} =  \Lambda_I[A,L]\Lambda_{J\setminus J_n} +  \Lambda_{I\setminus I_n}[A,L]\Lambda_J +  \Lambda_{I\setminus I_n}[A,L]\Lambda_{J\setminus J_n}
        }
        is a compact operator by assumption, since $\dist(I,J\setminus J_n),\dist(I\setminus I_n,J),\dist(I\setminus I_n,J\setminus J_n)$ are all strictly positive. Therefore, to show that $\Lambda_I[A,L]\Lambda_J$ is compact, it suffices to show that $\norm{\Lambda_{I_n}[A,L]\Lambda_{J_n}}\to 0$ using the fact the compact operators form a closed set. To this end, we consider
        \eq{
        \norm{\Lambda_{I_n}[A,L]\Lambda_{J_n}} &\leq \norm{\Lambda_{I_n}AL\Lambda_{J_n} - z_0 \Lambda_{I_n}A\Lambda_{J_n}} + \norm{z_0\Lambda_{I_n}A\Lambda_{J_n} - \Lambda_{I_n}LA\Lambda_{J_n}} \\
        &\leq \norm{A}\norm{L\Lambda_{J_n} - z_0 \Lambda_{J_n}} +  \norm{z_0\Lambda_{I_n} - \Lambda_{I_n}L}\norm{A} \\
        &\leq 2\max(\diam(I_n),\diam(J_n))\norm{A} \xrightarrow{ n \to \infty } 0
        }
        where we have used \cref{lem:bound on difference of spectral projections} right below for the last inequality.
        
        \paragraph{Step 3}
        Finally, we will show that $[A,L]$ is compact. Consider the dyadic partition of $\bS^1$, where at $n$-th partition, there are subintervals $I^n_k$ indexed by $1\leq k\leq 2^n$. In particular, we construct each $I^n_k$ to have one open and one closed end so that the collection $\{\Lambda_{I^n_k}\}_{k=1}^{2^n}$ consists of pairwise orthogonal projections. For any two disjoint subintervals $I^n_j,I^n_k$ for $j\neq k$ within each level ($n\geq 2$), either their closure contains a point, or they have a finite distance. In either case, we have $\Lambda_{I^n_j}[A,L]\Lambda_{I^n_k}\in\calK$. Let us define 
        \eq{
            T_n:=\sum_{k=1}^{2^n}\Lambda_{I^n_k}[A,L]\Lambda_{I^n_k}\,.
        } 
        We have 
        \eq{
            [A,L]-T_n& = \sum_{j=1}^{2^n} \Lambda_{I^n_j} [A,L] \sum_{k=1}^{2^n} \Lambda_{I^n_k} - \sum_{k=1}^{2^n}\Lambda_{I^n_k}[A,L]\Lambda_{I^n_k} \\
            &= \sum_{j\neq k}^{2^n}  \Lambda_{I^n_j} [A,L] \Lambda_{I^n_k} \in \calK
        }
        Therefore, to show that $[A,L]\in\calK$, it suffices to show that $\norm{T_n}\to 0$. Since the collection $\{\Lambda_{I^n_k}\}_{k=1}^{2^n}$ consists of pairwise orthogonal projections, we have
        \eq{
            \norm{T_n}=\max_{1\leq k\leq 2^n}\norm{\Lambda_{I^n_k}[A,L]\Lambda_{I^n_k}}\,.
        }
        We will now run an argument similar to the previous case. Let $z_0$ be any point in the interval $I^n_k$. Then
        \eq{
            \norm{\Lambda_{I^n_k}[A,L]\Lambda_{I^n_k}} & \leq \norm{\Lambda_{I^n_k}AL\Lambda_{I^n_k}-z_0\Lambda_{I^n_k}A\Lambda_{I^n_k}} + \norm{z_0\Lambda_{I^n_k}A\Lambda_{I^n_k}-\Lambda_{I^n_k}LA\Lambda_{I^n_k}} \\
            &\leq \norm{A}\norm{L\Lambda_{I^n_k}-z_0\Lambda_{I^n_k}} + \norm{z_0\Lambda_{I^n_k}-\Lambda_{I^n_k}L}\norm{A} \\
            &\leq 2\diam(I^n_k)\norm{A}
        }
        where we have used \cref{lem:bound on difference of spectral projections} on the last inequality again. Since $\diam(I^n_j)=\diam(I^n_k)$ for all $j,k$, it follows that
        \eq{
        \norm{T_n}\leq \max_{1\leq k\leq 2^n}2\diam(I^n_k)\norm{A}=2\diam(I^n_k)\norm{A}\to 0\,.
        }
        This completes the proof.

    \end{proof}

    \begin{lem}\label{lem:bound on difference of spectral projections}
    Let $A$ be a normal operator. Let $z\in\CC$ be a point and $S\subseteq \sigma(A)$ a subset such that $z$ lies in the closure $S$. Then \eq{\norm{A\chi_{S}(A)-z\chi_{S}(A)}\leq \operatorname{diam}(S)\,.}
    \end{lem}
    \begin{proof}
    It follows from
    \eq{
        \norm{A\chi_{S}(A)-z\chi_{S}(A)}\leq \sup_{\lambda\in\sigma(A)}|\lambda\chi_{S}(\lambda)-z\chi_{S}(\lambda)| =\sup_{\lambda\in S}|\lambda-z|\leq \operatorname{diam}(S)\,.
    }
    \end{proof}
    In conclusion, even though this characterization of $L$-locality only calls for the spectral projections, there is a \emph{shadow} of the $\ZZ^2$ geometry in the constraint about distances between intervals.
    \section{$L$-locality and infinitely many non-trivial projections}\label{sec:other forms of locality}
    Now that we have established an abstract characterization of $L$-locality via its spectral projections, we want to contrast it with other forms of locality, also defined via an infinite number of orthogonal projections, and ask whether these different forms are equivalent.
    
    Let us set up some terminology. We shall use the standard
    \begin{defn}[non-trivial projection]
        An orthogonal projection $P\in\calB$ is termed non-trivial iff both $\im P,\ker P$ are infinite dimensional subspaces. 
    \end{defn}

    Let us assume WLOG (via \cref{rem:spectral type of L}) that $L$ has a pure-point spectrum with each eigenvalue infinitely degenerate. Then its orthogonal eigenbasis spans the Hilbert space and the collection of projections onto each eigenvalue is a decomposition of the identity. Since each eigenvalue is infinitely degenerate, the associated projection is in fact non-trivial. Motivated by this, we have the abstract
    \begin{defn}[non-trivial decomposition of $\Id$]\label{defn:non-trivial decomposition of identity} Let $\Set{\Lambda_j}_{j\in \calJ}$ be a set of orthogonal projections indexed by a (finite or infinite) countable set $\calJ$ such that: \begin{enumerate}
        \item For each $j\in\calJ$, the projection $\Lambda_j$ is non-trivial.
        \item The projections are pairwise orthogonal: $\Lambda_j \Lambda_j = 0$ if $j\neq k$.
        \item The projections add up to the identity \eql{\Id = \sum_{j\in\calJ}\Lambda_j\,.}
    \end{enumerate}
    \end{defn}

	Given a fixed  decomposition of $\Id$, $\Set{\Lambda_j}_{j\in\calJ}$, and an operator $A$, we shall throughout use the shorthand $A_j$ for the operator restricted to the $j$th block, i.e., \eq{A_j\equiv\Lambda_j A \Lambda_j : \im(\Lambda_j)\to\im(\Lambda_j)\,.}
 
    \subsection{Other plausible notions of locality in two-dimensions}
    Motivated by the above abstract characterization, it is tempting to ask whether there are other, simpler characterizations of $L$-locality, and if so, are they equivalent, stronger or weaker. In the process of searching for such, we have come up with various strictly different notions.

    In the sequel, we let $\calJ$ be a countable indexed set and let $\Set{\Lambda_j}_{j\in\calJ}$ be a fixed abstract non-trivial decomposition of $\Id$. When $\calJ$ is countable, then it is bijective with the set of points on the circle with rational angles, $\bS^1_\QQ\cong[0,2\pi)\cap\QQ$. In that case, we may furnish $\calJ$ with a notion of \emph{distance} by pulling back the distance from the circle under this bijection. We shall use that notion in the sequel. Moreover, this bijection also allows one to refer to an "interval" within $\calJ$.  Let us introduce the "super" operator $\bbDelta:\calB\to\calB$ given by \eql{ A \mapsto \bbDelta(A) := \sum_{j\in\calJ}\Lambda_j A \Lambda_j\,.}
    It is well-known that $\bbDelta$ is continuous and maps compact operators to compact operators.

    \begin{defn}[alternative locality]
    An operator $A\in\calB$ is termed:
    \begin{itemize}
    \item \emph{$\Lambda$-local of Type I} iff $[A,\Lambda_j]\in\calK$ for all $j\in\calJ$.
    \item \emph{$\Lambda$-local of Type II} iff $A-\bbDelta(A)\in\calK$.
    \item \emph{$\Lambda$-local of Type III} iff $[A,\Lambda_S]\in\calK$ for all $S\subseteq\calJ$, where $\Lambda_S := \sum_{j\in S}\Lambda_j$.
    \item \emph{$\Lambda$-local of Type IV} iff $\Lambda_I A \Lambda_J\in\calK$ for all intervals $I,J\subseteq\calJ$ with $\dist(I,J)>0$ where $\Lambda_I:=\sum_{j\in I}\Lambda_j$.
    \item \emph{$\Lambda$-local of Type V} iff $\Lambda_I A \Lambda_J\in\calK$ for all intervals $I,J\subseteq\calJ$ that are disjoint, or equivalently, iff $[A,\Lambda_J]\in\calK$ for all intervals $J\subseteq \calJ$. 
    \end{itemize}
    We will denote the space of operators obeying this condition $\calL_{\sharp}$ with $\sharp$ the type, omitting $\Lambda$ from the notation since it is fixed throughout. We will write $\calL$ when $\sharp$ is clear from the context.
    \end{defn}

    \cref{fig:flow chart localities} shows some of the implications between various types of localities, and we prove them in
    \begin{itemize}
    \item $\boxed{\mathrm{II}}\Rightarrow\boxed{\mathrm{III}}$ and $\boxed{\mathrm{II}}\Rightarrow\boxed{\mathrm{V}}$: Let $A\in\calL_\mathrm{II}$ so $A-\bbDelta A =: K \in \calK$ and $S\subseteq\calJ$ be any set, then $[A,\Lambda_S]=[\bbDelta A+K,\Lambda_S]=[K,\Lambda_S]\in\calK$. Thus $A\in\calL_\mathrm{III}$. Let $I,J\subseteq\calJ$ be two disjoint intervals. Then $\Lambda_I(\bbDelta A+K)\Lambda_J=\Lambda_IK\Lambda_J\in\calK$. Thus $A\in\calL_\mathrm{V}$.
    \item $\boxed{\mathrm{III}}\Rightarrow\boxed{\mathrm{I}}$ and $\boxed{\mathrm{V}}\Rightarrow\boxed{\mathrm{IV}}$: These are clear from the definition.
    \item $\boxed{\mathrm{III}}\Rightarrow\boxed{\mathrm{V}}$: Let $I,J$ be two disjoint intervals. Since $I\subseteq J^c$, we have $\Lambda_I\Lambda_J^\perp=\Lambda_I$. Let $A\in\calL_\mathrm{III}$, then we have $[A,\Lambda_J]\in\calK$. In particular, $\Lambda_I\Lambda_J^\perp[A,\Lambda_J]=\Lambda_I\Lambda_J^\perp A\Lambda_J=\Lambda_IA\Lambda_J\in\calK$. Thus $A\in\calL_\mathrm{V}$.
    \end{itemize}

    \begin{rem}[Distinction between $\calL_\mathrm{I}$ and $\calL_\mathrm{II}$]\label{rem:difference between LI and LII}
    The commutator $[A,\Lambda_j]$ is compact iff $\Lambda_j^\perp A\Lambda_j$ and $\Lambda_j A\Lambda_j^\perp$ are both comapct. Indeed, if $[A,\Lambda_j]\in\calK$, then $\Lambda_j^\perp[A,\Lambda_j]=\Lambda_j^\perp A\Lambda_j$ and $-[A,\Lambda_j]\Lambda_j^\perp = \Lambda_jA\Lambda_j^\perp$ are compact; conversely, we have $[A,\Lambda_j]=(\Lambda_j+\Lambda_j^\perp)[A,\Lambda_j] = \Lambda_jA\Lambda_j-\Lambda_jA + \Lambda_j^\perp A\Lambda_j = -\Lambda_jA\Lambda_j^\perp+\Lambda_j^\perp A\Lambda_j$.

    On the other hand, for $A\in\calL_\mathrm{II}$, it says that $A-\bbDelta A \in\calK$ and we have $A-\bbDelta A = \sum_{j\in\calJ}\Lambda_j A\sum_{j\in\calJ}\Lambda_j - \sum_{j\in\calJ}\Lambda_jA\Lambda_j=\sum_{j\neq k}\Lambda_jA\Lambda_j$. In particular \eq{
        \sum_{j\in\calJ}\Lambda_jA\Lambda_j^\perp = \sum_{j\in\calJ}\Lambda_jA\sum_{k\neq j}\Lambda_k = A-\bbDelta A  = \sum_{j\in\calJ}\sum_{k\neq j}\Lambda_kA\Lambda_j = \sum_{j\in\calJ}\Lambda_j^\perp A\Lambda_j\,.
    }
    Therefore, for $\calL_\mathrm{I}$, each infinite block $\Lambda_jA\Lambda_j^\perp$ and $\Lambda_j^\perp A\Lambda_j$ is compact separately, and for $\calL_\mathrm{II}$, the operator comprised of infinitely many infinite blocks is compact. Moreover, if there are only finitely many blocks, then of course the two notions coincide.

    \end{rem}

    \begin{lem}[C-star algebraic structure]\label{lem:C-star algebraic structure}
        Each locality space defined above is furnished with the structure of a C-star algebra.
    \end{lem}
    \begin{proof}
        Let $\sharp\in\Set{\mathrm{I,III,V}}$ the locality conditions are of the form $[A,\Lambda_S]\in\calK$ for some $S\subseteq \calJ$. If $A,B\in\calL_\sharp$, then $[AB,\Lambda_S]=A[B,\Lambda_S]+[A,\Lambda_S]B\in\calK$. If $A_n\to A$ in operator norm and $A_n\in\calL_\sharp$, then $\calK\ni [A_n,\lambda_S]\to [A,\Lambda_S]$. Since operator norm limit of compact operators is compact, it follows that $[A,\Lambda_S]\in\calK$.

        For the space $\calL_{\mathrm{IV}}$, we use \cref{thm:abstract-L-locality} and the locality condition is equivalent to $[A,L]\in\calK$ for some unitary operator $L$. The arguments in the previous paragraph work in this case.

        Consider the space ${\calL_{\mathrm{II}}}$. If $A,B\in{\calL_{\mathrm{II}}}$, write $A=\bbDelta A+K_A$ and $B=\bbDelta B + K_B$ where $K_A,K_B\in\calK$, then $AB-\bbDelta(AB)=\bbDelta A K_B + K_A\bbDelta B + K_AK_B -\bbDelta(\bbDelta A K_B + K_A\bbDelta B + K_AK_B)\in\calK$ since $\bbDelta$ preserves compactness. If $A_n\in\calL_{\mathrm{III}}$ and $A_n\to A$ in norm, then $A_n-\bbDelta A_n\to A-\bbDelta A$ in norm, using the continuity of $\bbDelta$. Thus $A\in{\calL_{\mathrm{II}}}$.

        It follows that each locality space is a C-star algebra.
    \end{proof}
        
    \begin{figure}[t]
    \centering
    \begin{tikzcd}[color=black]
        \boxed{\mathrm{II}} \arrow[r,Rightarrow] \arrow[rd,Rightarrow] & \boxed{\mathrm{III}} \arrow[r,Rightarrow] \arrow[d,Rightarrow] & \boxed{\mathrm{I}}\\
        & \boxed{\mathrm{V}} \arrow[r,Rightarrow] & \boxed{\mathrm{IV}}
    \end{tikzcd}
    \caption{Implications of various types of locality.}
    \label{fig:flow chart localities}
    \end{figure}
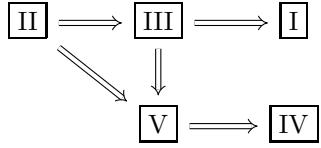

    It should be noted that when $\calJ$ is finite, when it makes sense (which is for Type-I/II/III), there is no difference between the various types.
    
    \begin{rem}[Relevance to two-dimensional locality in physics and the term quasi-2D-locality]
        Since locality with Type IV is equivalent to $L$-locality via \cref{thm:abstract-L-locality}, it would have been plausible that any other Type could have just as well been used to characterize locality in 2D, if we take $\Lambda_j$'s to be the spectral projections associated with the particular choice of Laughlin's $L$ in \cref{eq:Laughlin flux operator}. It will turn out, however, that calculating $\pi_0$ of Type-I and Type-II local operators starkly deviates from what we expect of Type-IV (corresponding to the Kitaev table) and as such we term such locality \emph{quasi-2D}.
    \end{rem}

    We shall only study $\pi_0$ of Type-I and Type-II local operators in the sequel. It will be useful to extract some common features of both before tackling this.
    
    Define $\calD\subseteq \calB$ to be the space of operators that commute with $\Lambda_j$ for all $j\in\calJ$. We can also describe Type-I/II local operators alternatively as:
    \begin{prop}\label{prop:quasi compcat}
    The spaces $\calL_{\mathrm{I,II}}$ may also be characterized as: \eq{
        \calL_\mathrm{I}=\calD+\calK_\mathrm{I}
    } where $\calK_\mathrm{I}$ is the space of operators $B\in\calB$ such that both $\Lambda_jB$ and $B\Lambda_j$ are compact for all $j\in\calJ$; and \eq{
        \calL_\mathrm{II}=\calD+\calK\,.
    }
    \end{prop}
    \begin{rem}[Another space consisting of sums of a diagonal plus a compact operator]
    Let $\Set{\vf_k}_{k\in \NN}$ be an ONB and $\ti\calD$ be the space of operators $D$ that are diagonalized by the basis, i.e., $D\in\ti\calD$ iff $D\vf_k=\lambda_k\vf_k$ for some $\lambda_k\in\CC$ for all $k\in\NN$. The space $\ti\calD +\calK$ is studied in detail by Andruchow, et al. \cite{andruchow2019c}. The space $\ti\calD$ coincides with $\calD$ if we have rank-one decomposition of $\Id$, i.e., the projections $\Lambda_j$ in \cref{defn:non-trivial decomposition of identity} are all rank-one projections. They show that the space of unitary operators in $\ti\calD+\calK$ is path-connected \cite[Proposition 3.6]{andruchow2019c}, and that the space of non-trivial projections in $\ti\calD+\calK$ has uncountably many path-connected components \cite[Corollary 5.10]{andruchow2019c}. In contrast, we will show that the space of $\Lambda$-nontrivial projections in $\calL_\mathrm{II}\equiv \calD+\calK$ is path-connected.
    \end{rem}
    \begin{proof}
    For $A\in\calL_\mathrm{I}$, let $D=\bbDelta(A)$ and we write $A=D+(A-D)$. Then $D\in\calD$ and $A-D\in\calK_\mathrm{I}$. Indeed, we have $(A-D)\Lambda_j = A\Lambda_j-\Lambda_jA\Lambda_j = \Lambda_j^\perp A\Lambda_j$ and $\Lambda_j(A-D) = \Lambda_j A\Lambda_j^\perp$; and $[A,\Lambda_j]\in\calK$ iff both $ \Lambda_j^\perp A\Lambda_j$ and $ \Lambda_j A\Lambda_j^\perp$ are compact. Conversely, if $D+B\in\calD+\calK_\mathrm{I}$, then $[D+B,\Lambda_j]=[B,\Lambda_j]\in\calK$. The fact that $\calL_\mathrm{II}=\calD+\calK$ is clear from definition.
    \end{proof}

    Clearly $\calK$ is an ideal in $\calB$, but also, $\calK$ is an ideal within the C-star algebra $\calD+\calK$ as one may verify. Analogously, the subspace $\calK_\mathrm{I}$ is an ideal within the C-star algebra $\calD+\calK_\mathrm{I}$. Indeed, if $D+B\in\calD+\calK_\mathrm{I}$ and $T\in\calK_\mathrm{I}$, then $(D+B)T\Lambda_j\in\calK$ and $\Lambda_j(D+B)T = D\Lambda_jBT + \Lambda_j BT\in\calK$, and similarly $T(D+B)\in\calK_\mathrm{I}$.

    For the spaces $\calK$ and $\calK_\mathrm{I}$, we have the following
    \begin{lem}\label{lem:compact operators}
    Let $\ve>0$ be given. 
    \begin{enumerate}
    \item If $K\in\calK$, then there exists $m<\infty$ and $P=\sum_{j=1}^m\Lambda_j$ such that $\norm{K-PKP}<\ve$. 
    \item If $A\in\calK_\mathrm{I}$, then there exists a projection $P$ such that $\norm{A-PAP}< \ve$ and $PTP\in\calL_\mathrm{I}$ for all $T\in\calB$. (More explicitly, for a given $\{\vf^j_k\}_{k=1}^\infty$ orthonormal basis on each $\im\Lambda_j$, if we denote $E^j_k=\vf^j_k\otimes (\vf^j_k)^\ast$, then the projection is of the form $
        P=\sum_{j=1}^\infty \sum_{k=1}^{n_j} E_k^j
    $ for some $n_j\in\NN$.)
    \end{enumerate}
    \end{lem}
    \begin{proof}
    If $K\in\calK$, then there exists $m$ large enough such that 
    \eq{
        \max\left(\norm{\sum_{j=m+1}^\infty\Lambda_jK},\norm{\sum_{j=m+1}^\infty K\Lambda_j}\right)<\ve\,.
    }
    Let $P = \sum_{j=1}^m\Lambda_j$. We have $\norm{K-PKP}<3\ve$.

        Let $\{\vf^j_k\}_{k=1}^\infty$ be an orthonormal basis for $\im\Lambda_j$ and let $E^j_k=\vf^j_k\otimes (\vf^j_k)^\ast$. Since $\Lambda_jA$ and $A\Lambda_j$ are compact, there exists $n_j$ large enough such that
    \eq{
        \max\left(\norm{\sum_{k=n_j+1}^\infty E_k^j A},\norm{\sum_{k=n_j+1}^\infty  A E_k^j}\right) < \ve/2^j\,.
    }
    Define $P=\sum_{j=1}^\infty \sum_{k=1}^{n_j} E_k^j$. We have
    \eq{
        \norm{A-PAP} \leq \norm{PAP^\perp}+\norm{P^\perp AP}+\norm{P^\perp AP^\perp}< 3\ve\,.
    }
    For arbitrary $T\in\calB$, since each $\sum_{k=1}^{n_j}E_k^j$ is finite-rank, we have that
    \eq{
        [PTP,\Lambda_j]=PTP\Lambda_j-\Lambda_jPTP =\sum_{k=1}^{n_j} (PT E_k^j - E_k^jTP)
    }
    is finite-rank.
    \end{proof}

    \section{Classification of quasi-2D unitaries}\label{sec:classification of quasi2D unitaries}
    Let $\calU^\calL:=\calU\cap\calL$. Denote $\calG$ as the space of invertible operators. Let $\calG^\calL:=\calG\cap\calL$. We define the super operator $\bbLambda_j:\calB\to\calB$ via $A\mapsto \bbLambda_j A := \Lambda_jA\Lambda_j +\Lambda_j^\perp\cong A_j\oplus \Id_{\im(\Lambda_j^\perp)}$.
    \begin{thm}\label{thm:pi0 of quasi-2D unitaries} The set of path-connected component for $\calU^{\calL_\mathrm{I}}$ is
        \eql{ \pi_0(\calU^{\calL_\mathrm{I}}) \cong \prod_{j=1}^{|\calJ|-1} \ZZ }
        the bijection given by \eql{\label{eq:bijection with direct product}
            \calU^{{\calL_{\mathrm{I}}}}\ni U\mapsto \Set{\findex \bbLambda_j U}_{j\in\calJ}\,,
        } and the set of path-connected component for $\calU^{\calL_\mathrm{II}}$ is
        \eql{ \pi_0(\calU^{\calL_\mathrm{II}}) \cong \bigoplus_{j=1}^{|\calJ|-1} \ZZ } the bijection again given by \eql{\label{eq:bijection with direct sum}
            \calU^{{\calL_{\mathrm{II}}}}\ni U\mapsto \Set{\findex \bbLambda_j U}_{j\in\calJ}\,,\qquad|\Set{j\in\calJ|\findex \bbLambda_j U\neq0}|<\infty\,.
        } 
        
        I.e., if $|\calJ|=\infty$, then the latter is to be interpreted as the space of sequences $\calJ\to\ZZ$ with finitely-many non-zero values and the former is general sequences $\calJ\to\ZZ$.

    \end{thm}
        Clearly when $|\calJ|=2$ we are back to the usual one-dimensional notion of $\Lambda$-locality from \cite{ChungShapiro2023}, with $\Lambda_1 := \Lambda$ and $\Lambda_2 := \Lambda^\perp$.
    \begin{proof}
        Let $U\in\calU^{\calL_\mathrm{I}}$. We first argue that if $\findex \bbLambda_j U=0$ for each $j$, then $U\leadsto \Id$ within $\calG^{\calL_\mathrm{I}}$. Recall we use the notation $U_j$ for the map $U_j:\im(\Lambda_j)\to\im(\Lambda_j)$ defined via $U_j\psi:=\Lambda_jU\Lambda_j\psi$. By locality and unitarity, it is clear that $U_j$ is Fredholm with parametrix $U_j^\ast$, and $\findex U_j = \findex \bbLambda_j U = 0$ by hypothesis. Thus we have $U_j=G_j+K_j$ for some $G_j\in\calG$ and $K_j\in\calK$. Let $G$ be the direct sum of $G_j$. Then $U=G+A$, where $A\equiv U-G$ is \emph{not} necessarily compact, but has the property that $\Lambda_j A,A\Lambda_j\in\calK$ (so it's in $\calK_I$). Indeed, we have $A=U-G\in{\calL_\mathrm{I}}$, which implies $\Lambda_j^\perp A\Lambda_j,\Lambda_j A\Lambda_j^\perp \in\calK$. It follows from $\Lambda_jA\Lambda_j=K_j\in\calK$ that $\Lambda_jA=\Lambda_jA\Lambda_j^\perp +\Lambda_jA\Lambda_j\in\calK$ and similarly for $A\Lambda_j\in\calK$. Let $B=AG^{-1}$. We have
        \eq{
            U = G+A = (\Id +AG^{-1})G \leadsto \Id + B
        }
        where we deform each $G_j\leadsto \Id$ within $\calG$ (using Kuiper), which gives a path $G\leadsto \Id$ that commutes with $\Lambda_j$, and hence is local. It also follows from $[G^{-1},\Lambda_j]=0$ that $\Lambda_jB,B\Lambda_j\in\calK$. Indeed, $B\Lambda_j=AG^{-1}\Lambda_j = A\Lambda_j G^{-1}\in\calK$. So $B\in\calK_I$ in fact.

        Let $\ve>0$. Using \cref{lem:compact operators}, there exists a projection $P$ such that $\norm{B-PBP}\leq \ve$ and $PTP\in\calL_\mathrm{I}$ for all $T\in\calB$. If we choose $\ve$ small enough, using the openness of the space of inveritble operators, we can deform $\Id+B\leadsto \Id +PBP$ within $\calG^{\calL_\mathrm{I}}$ via straight-line homotopy. If we write $\Id+PBP=P^\perp + P(\Id+B)P$, we see that $P(\Id+B)P$ is invertible on $\im P$. We can then deform $P(\Id+B)P\leadsto \Id_{\im P}$ within $\calG$. In particular, this path will not violate type-I locality since $[PTP,\Lambda_j]\in\calK$ for all $T\in\calB$. Therefore, we have deformed $\Id+PBP\leadsto \Id$ within $\calG^{\calL_\mathrm{I}}$. We can lift the path $U\leadsto \Id$ in $\calG^{\calL_\mathrm{I}}$ to a path in $\calU^{\calL_\mathrm{I}}$ using \cref{lem:lift invertible local path to unitary local path}.

    Before completing the analysis of ${\calL_{\mathrm{I}}}$, we proceed to consider ${\calL_{\mathrm{II}}}$. Let $U\in\calU^{{\calL_\mathrm{II}}}$. We first argue that all but finitely many $\findex\bbLambda_j U$ must vanish. Write $U=\bbDelta U+K$ with $K\in\calK$. Since $K$ is compact, using \cref{lem:compact operators}, there is $m$ large enough and $P=\sum_{j=1}^m\Lambda_j$ such that $\norm{K-PKP}<\ve$. If we choose $\ve$ small enough, then by the openness of the space of invertible operators, we can deform $U=\bbDelta U+K\leadsto \bbDelta U+PKP$ in $\calG^{{\calL_\mathrm{II}}}$ via straight-line homotopy. Since $\bbDelta U$ commutes with $\Lambda_j$, we can write
    \eq{
        \bbDelta U+PKP = P^\perp \bbDelta U P^\perp + P(\bbDelta U+K)P\,.
    }
    Thus $U_j\in\calG$ for $j\geq m+1$ (we use the same definition of $U_j$ as in the previous paragraph) and hence $\findex\bbLambda_j U=\findex U_j=0$ for $j\geq m+1$, i.e., all but finitely many of the indices are zero. 
    
    We now deform $U\in\calU\cap{\calL_{\mathrm{II}}}$ to $\Id$ within $\calU\cap{\calL_{\mathrm{II}}}$. Initially, as before, we assume that \emph{all} indices are zero. For the operator $P^\perp \bbDelta U P^\perp$, since $U_j\in\calG$ for $j\geq m+1$, we deform each $U_j\leadsto \Id$ within $\calG$ (this path is local as it is diagonal). Next, the operator $P(\bbDelta U+K)P$ belongs to $\calG^{\calL_\mathrm{I}}$ on $\im P$, because there are only finitely many projections in the decomposition $\Id_{\im P}=\sum_{j=1}^m\Lambda_j$, see \cref{rem:difference between LI and LII}. Since all indices are assumed to be zero, we can use the previous result to deform $P(\bbDelta U+K)P\leadsto \Id_{\im P}$ within $\calG^{\calL_\mathrm{I}}$ which coincides with $\calG^{{\calL_\mathrm{II}}}$ in this case (again using \cref{rem:difference between LI and LII}). The combined path gives $\bbDelta U+PKP\leadsto \Id$ within $\calG^{{\calL_\mathrm{II}}}$. We can lift the path to $\calU^{{\calL_\mathrm{II}}}$ using \cref{lem:lift invertible local path to unitary local path}.

    We have thus far established that the zero index set is path-connected, for $\calU^{\calL_\sharp}$ both with $\sharp\in \Set{\mathrm{I,II}}$. To deal with non-zero indices, suppose $U,V\in\calU^{\calL_{\sharp}}$ have the same indices $\findex\bbLambda_jU = \findex\bbLambda_jV$ for each $j$. Then $UV^*\in\calU^{\calL_\sharp}$ has $\findex\bbLambda_j UV^* = 0$ for each $j$. Thus we have a path $\gamma_t$ that deforms $UV^\ast\leadsto \Id$ within $\calU^{\calL_\sharp}$. In particular, $\gamma_tV$ deforms $U\leadsto V$ within $\calU^{\calL_\sharp}$. This completes the proof that the maps \cref{eq:bijection with direct product,eq:bijection with direct sum} are injective.

    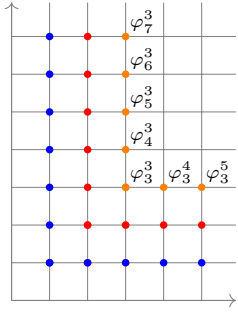
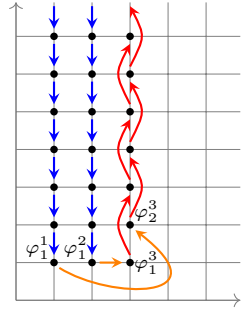
\begin{figure}[t]    
    \centering
    \begin{subfigure}{0.4\linewidth}
        \centering
        \begin{tikzpicture}[scale=0.5]
        \draw[step=1cm,gray,very thin] (-1,-1) grid (4.9,6.9);
        \draw[->,gray] (-1,-1) -- (4.9,-1);
        \draw[->,gray] (-1,-1) -- (-1,6.9);
        
        \def\colors{{"blue","red","orange","teal"}};
        \foreach \x in {0,1,2}{
        \foreach \y in {0,1,2,3,4,5,6}{
        \pgfmathsetmacro{\c}{\colors[\x]}
        \node at (\x,\y) [circle,fill,inner sep=1pt,\c]{};
        }
        }
        
        \foreach \y in {0,1,2}{
        \foreach \x in {\y,...,4}{
        \pgfmathsetmacro{\c}{\colors[\y]}
        \node at (\x,\y) [circle,fill,inner sep=1pt,\c]{};
        }
        }
        \node at (2,3) [above right, yshift=-2,xshift=-2,black]{$\scriptstyle \varphi^3_4$};
        \node at (2,4) [above right, yshift=-2,xshift=-2,black]{$\scriptstyle\varphi^3_5$};
        \node at (2,5) [above right, yshift=-2,xshift=-2,black]{$\scriptstyle\varphi^3_6$};
        \node at (2,6) [above right, yshift=-2,xshift=-2,black]{$\scriptstyle\varphi^3_7$};
        
        \node at (2,2) [above right, yshift=-2,xshift=-2,black]{$\scriptstyle\varphi^3_3$};
        \node at (3,2) [above right, yshift=-2,xshift=-2,black]{$\scriptstyle\varphi^4_3$};
        \node at (4,2) [above right, yshift=-2,xshift=-2,black]{$\scriptstyle\varphi^5_3$};

        \end{tikzpicture}
        \caption{Illustration of the decomposition $\calH=\oplus_{n\in\NN}\calH_n$ and its natural ordering. Each column $j$ sits the orthonormal basis that spans $\im\Lambda_j$. The blue basis elements span $\calH_1$, the red spans $\calH_2$, and the orange spans $\calH_3$. Some of the basis elements for $\calH_3$ are displayed. Each $\calH_n$ is naturally isomorphic to $\ell^2(\ZZ)$, with the vertical part of $\calH_n$ corresponding to the left part of $\ell^2(\ZZ)$, and the horizontal to the right part. On each $\calH_n$, one can therefore define a  bilateral shift.}
        \label{fig:surjectivity infinite calJ}
    \end{subfigure}
    \qquad
    \begin{subfigure}{0.4\linewidth}
    \centering
        \begin{tikzpicture}[scale=0.5]
        \draw[step=1cm,gray,very thin,opacity=0.5] (-1,-1) grid (4.9,6.9);
        \draw[->,gray] (-1,-1) -- (4.9,-1);
        \draw[->,gray] (-1,-1) -- (-1,6.9);
        
        \def\colors{{"blue","red","orange","teal"}};
        \foreach \x in {0,1,2}{
        \foreach \y in {0,1,2,3,4,5,6}{
        \node at (\x,\y) [circle,fill,inner sep=1pt,black]{};
        }
        }
        
        \foreach \y in {1,...,7}{
        \draw[-stealth,blue,thick] (0,\y-0.2) -- (0,\y-1+0.2);
        }
        \foreach \y in {1,...,7}{
        \draw[-stealth,blue,thick] (1,\y-0.2) -- (1,\y-1+0.2);
        }
        \foreach \y in {7,5,3}{
        \draw[stealth-,red,thick] (2,\y-0.2) .. controls (2+0.4,\y-1) .. (2,\y-2+0.2);
        }
        \foreach \y in {6,4,2}{
        \draw[stealth-,red,thick] (2,\y-0.2) .. controls (2-0.4,\y-1) .. (2,\y-2+0.2);
        }
        
        \draw[-stealth,orange,thick] (0+0.14,0-0.14) .. controls (1.5,-1) and (4.5,-1)  .. (2+0.14,1-0.14);
        \draw[-stealth,orange,thick] (1+0.2,0) -- (2-0.2,0);
        
        \node at (0,0) [above left, yshift=-2,xshift=2,black]{$\scriptstyle\varphi^1_1$};
        \node at (1,0) [above left, yshift=-2,xshift=2,black]{$\scriptstyle\varphi^2_1$};
        \node at (2,0) [right, yshift=0,xshift=-2,black]{$\scriptstyle\varphi^3_1$};
        \node at (2,1) [above right, yshift=-2.5,xshift=-2,black]{$\scriptstyle\varphi^3_2$};
        
        \end{tikzpicture}
        \caption{Example of an operator in $\calU^{\calL_\mathrm{I}}$ when $|\calJ|=3$ and having indices $1,1,-2$. In each column sits the basis elements that spans $\im\Lambda_j$. We define unilateral shifts according the the indices. Here $\vf_1^1,\vf_1^2$ are initially mapped to zero, while $\vf_1^3,\vf_2^3$ are not mapped to. The orange arrows modify the operator in such a way that makes it unitary.}
        \label{fig:surjectivity finite calJ}
    \end{subfigure}
    \caption{Construction of operators in $\calU^\calL{_{\mathrm{I}}}$.}
    \end{figure}

    We proceed to show surjectivity. We shall first reduce the surjectivity in ${\calL_{\mathrm{II}}}$ to that in ${\calL_{\mathrm{I}}}$. To do so, let a sequence $s:\calJ\to\ZZ$ be given so that only finitely many of its elements are non-zero. Then we decompose $\calJ=\calJ_0\sqcup\calJ_1$ where the former set is the infinite of zero indices. This induces a direct sum decomposition on the Hilbert space \eq{
        \left(\bigoplus_{j\in\calJ_0}\im\Lambda_j\right)\oplus\left(\bigoplus_{j\in\calJ_1}\im\Lambda_j\right) =: \calH_0 \oplus \calH_1\,.
    } Then it is clear that to find an operator whose index is $s$, we can set it equal to $\Id$ on $\calH_0$, and merely look for an operator on $\calH_1$ whose index is $s$. However, since $|\calJ_1|<\infty$, this problem, by \cref{rem:difference between LI and LII} reduces to finding an operator within ${\calL_{\mathrm{I}}}$. Thus, let us show that each path-connected component of $\calU^{\calL_\mathrm{I}}$ is nonempty. 

    We first consider when $\calJ$ is an infinite set so WLOG $\calJ=\NN$ and let $\Set{s_j}_{j\in\NN}\subseteq\ZZ$. Let $\Set{\vf_k^j}_{k\in\NN}$ be an ONB for $\im\Lambda_j$ for each $j\in\NN$. Define "L"-shaped subspace \eq{
        \calH_n = \szpan\Set{\vf^n_k,\vf^j_n | k\geq n,j>n},\qquad (n\in\NN)\,.
    }
    For $n\neq m$, the subspaces $\calH_n$ and $\calH_m$ are orthogonal since they contain distinct basis elements, and we have $\calH=\oplus_{n\in\NN}\calH_n$. There is a natural ordering (in $\ZZ$) of elements in $\calH_n$: \eq{
    \dots,\vf^n_{n+1},\vf^n_n,\vf^{n+1}_n,\vf^{n+2}_n,\dots
    }
    which makes $\calH_n$ isomorphic with $\ell^2(\ZZ)$ in a way that respects ${\calL_{\mathrm{I}}}$ locality; see \cref{fig:surjectivity infinite calJ} for an illustration. With respect to this ordering, on $\calH_n\cong\ell^2(\ZZ)$, we consider the operator $R_n:=R^{s_n}$ where $R$ is the bilateral right shift. We define the unitary operator $U=\oplus_{n\in\NN}R_n$ on $\calH$. It is clear that $U\in\calL_\mathrm{I}$. In fact, there appears only finite hopping between basis elements in different $\im\Lambda_j$'s, and therefore $\Lambda_j U\Lambda_j^\perp$ and $\Lambda_j^\perp U\Lambda_j$ are finite-rank. It is clear by construction that $\findex \bbLambda_j U = s_j$.

    We now consider the case when $\calJ$ is finite. Let $m:=|\calJ|$ and let $\Set{s_j}_{j=1}^{m}\subseteq \ZZ$ be a finite set of integers such that \eql{
        s_1+\cdots+s_m=0\,. \label{eq:finite calJ index constraint}
    }
    In general, on $\im\Lambda_j$, we define $\hat{R}$ to be the unilateral shift operator \eq{
        \hat{R}\vf^j_k = \begin{cases}
            \vf^j_{k-1} & k > 1 \\
            0 & k\leq 1
        \end{cases}\,.
    }
    We consider the operator $\hat{R}^{s_j}$ on $\im\Lambda_j$ so that $\findex \hat{R}^{s_j} = s_j$. If $s_j>0$, then the elements in $\Set{\vf^j_1,\dots,\vf^j_{s_j}}$ are mapped to zero; while if $s_j<0$, then those elements are \emph{not} mapped to. The operator $\bigoplus_{j=1}^m \hat{R}^{s_j}$ is in $\calL_\mathrm{I}$ but not unitary. To make it unitary, we can modify its operation between elements in $\bigcup_{j=1}^m \Set{\vf^j_1,\dots,\vf^j_{s_j}}_j$. Because of the constraint \cref{eq:finite calJ index constraint}, the number of elements mapping to zero is equal to the number of elements that are not mapped to. We will map those elements initially mapping to zero to those elements that are not mapped to. Therefore, the operator so constructed becomes unitary. See \cref{fig:surjectivity finite calJ} for an example.

    \end{proof}

    \begin{lem}\label{lem:lift invertible local path to unitary local path}
    Let $U,V\in\calU^\calL$ and suppose there is a path $U\leadsto V$ within $\calG^\calL$. Then there is a path $U\leadsto V$ within $\calU^\calL$.
    \end{lem}
    \begin{proof}
    Let $\gamma_t\in\calG^\calL$ deforming $U\leadsto V$. Since $\calL$ is a C-star algebra, and $z\mapsto |z|$ is continuous, it follows that $|\gamma_t|\in\calL$. Then $\polar(\gamma_t)=\gamma_t|\gamma_t|^{-1}\in\calL$. Therefore, the path $t\mapsto \polar(\gamma_t)$ deforms $U\leadsto V$ within $\calU^\calL$.
    \end{proof}

    \section{Classification of quasi-2D self-adjoint unitaries}\label{sec:classification of quasi 2D SAUs}
    We denote the space of all self-adjoint unitaries (SAUs henceforth) by $\calSU$ (in \cite{ChungShapiro2023} the symbol for this set was $\calS$ but since here we shall refer also to self-adjoint invertibles we make this modification). A SAU $U$ is termed non-trivial iff $\dim\ker (U\pm \Id)$ are infinite-dimensional, i.e., $\sigmaess(U)=\Set{\pm 1}$. Denote the space of all self-adjoint invertible operators (SAIs henceforth) by $\calSG$. Define the $\sgn:\RR\to\RR$ function via \eq{
        \sgn(\lambda) := \begin{cases}
            1 & \lambda > 0 \\
            -1 & \lambda < 0 \\
            0 & \lambda = 0
        \end{cases}\,.
    }  
    Clearly if $G$ is a SAI then $\sgn(G)$ is a SAU. A SAI $G$ is termed non-trivial iff $\sgn(G)$ is non-trivial.
    
    As before in \cite[Section 4]{ChungShapiro2023}, we focus on $\Lambda$-non-trivial SAUs, but in the present context, $\Lambda$ is not a single projection operator, but rather, a fixed non-trivial decomposition of $\Id$ as in the foregoing sections. To that end, let $\pi$ be the $\ast$-homomorphism from $\calB(\calH)$ to the Calking algebra. We define
    \begin{defn}[$\Lambda$-non-trivial operators]\label{def:Lambda NT} 
    Let $A\in\calB(\calH)$ and $A_j:=\Lambda_jA\Lambda_j:\im\Lambda_j\to\im\Lambda_j$. If $A_j$ is self-adjoint and Fredholm, then $A$ is termed $\Lambda$-non-trivial iff 
    \eql{\label{eq:Lambda NT}
        \sigma(\sgn(\pi(A_j)))=\Set{\pm 1},\qquad (j\in\calJ)\,.
    }
    \end{defn}

    \begin{rem}\label{rem:local SAI Fredholm}
    Let $\sharp\in\Set{\mathrm{I,II}}$. If $G\in\calSG\cap \calL_\sharp$, then $G_j$ is self-adjoint and Fredholm. Indeed, we have \eq{
        (\Lambda_j G\Lambda_j) (\Lambda_j G^{-1}\Lambda_j) -\Lambda_j = \Lambda_j[G,\Lambda_j]G_j^{-1}\Lambda_j + \Lambda_j - \Lambda_j \in\calK 
    }
    and hence $G_jG_j^{-1}-\Id_{\im\Lambda_j}\in\calK$; similarly $G_j^{-1}G_j-\Id_{\im\Lambda_j}\in\calK$. Therefore, it makes sense to talk about $\Lambda$-non-triviality for operators in $\calSG\cap \calL_\sharp$. Thus, for $\sharp\in\Set{\mathrm{I,II}}$, we will denote $\calSG^{\calL_\sharp}_\LamNT$ to be the space of operators in $\calSG\cap\calL_\sharp$ that are $\Lambda$-non-trivial. We also let $\calSU^{\calL_\sharp}_\LamNT$ to be the space of operators in $\calSU\cap\calL_\sharp$ that are $\Lambda$-non-trivial. We let $\calSG^\calL:=\calSG\cap\calL$ and $\calSU^\calL:=\calSU\cap\calL$.
    \end{rem}
    
    \begin{rem}[Comparison with the notion of $\Lambda$-non-triviality in \cite{ChungShapiro2023}]
        It may be worthwhile to pause briefly and comment on how this notion of $\Lambda$-non-triviality differs from that presented in \cite{ChungShapiro2023}. There, $\Lambda$ is a \emph{single, fixed} non-trivial projection and a SAU $U$ was defined as $\Lambda$-non-trivial iff there were a SAU $V$ which has the following three properties: (1) $[\Lambda,V]=0$, (2) $\Lambda V \Lambda:\im(\Lambda)\to\im(\Lambda), \Lambda^\perp V \Lambda^\perp:\im(\Lambda^\perp)\to\im(\Lambda^\perp)$ are both non-trivial SAUs, and (3) $U-V\in\calK$. 
        
        This third condition automatically implies that a $\Lambda$-non-trivial SAU is $\Lambda$-local, since it says two things: both that the off-diagonal elements $\Lambda U \Lambda^\perp$ are compact (locality) and that the diagonal elements $\Lambda U \Lambda, \Lambda^\perp U \Lambda^\perp$ are compactly away from a non-trivial element, and are hence themselves non-trivial. That double-statement was a convenient way, there, to phrase $\Lambda$-non-triviality.

        Here, however, since we have various different notions of locality, it makes sense to disentangle non-triviality from locality and hence, we emphasize, that \cref{def:Lambda NT} does \emph{not} imply locality.

        In analogy, to disentangle the definition given in \cite{ChungShapiro2023} one would make the following distinction:

        A SAU $U$ is $\Lambda$-non-trivial iff \eq{\sigmaess(\Lambda^\sharp U \Lambda^\sharp:\im(\Lambda^\sharp)\to\im(\Lambda^\sharp))=\Set{\pm1}\qquad(\sharp\in\Set{\text{nothing},\perp})
        } and as expected, a SAU $U$ is $\Lambda$-local iff $[U,\Lambda]\in\calK$. The two together hold iff the three conditions above do.
        
    \end{rem}

    \begin{rem}\label{rem:non-triviality when A_j is invertible}
        Let $A\in\calB(\calH)$ be an operator such that $A_j$ is self-adjoint and Fredholm. Suppose $A$ is $\Lambda$-non-trivial. If $A_j$ is, in fact, invertible, then we have \eq{
            \sigmaess(\sgn(A_j)) = \Set{\pm 1}\,.
        }
        I.e., $A_j$ is non-trivial. Indeed, since $A_j$ is invertible, the $\sgn$ function becomes continuous on $\sigma(A_j)$, and we have \eq{
            \sigmaess(\sgn(A_j)) \equiv \sigma(\pi(\sgn(A_j))) = \sigma(\sgn(\pi (A_j))) = \Set{\pm 1}
        }
        where we have used \cite[Prop. 4.4.7]{kadison_fundamentals_1997} in the third equality and the $\Lambda$-non-triviality on the last equality.
    \end{rem}

    \begin{thm}\label{thm:pi0 of quasi-2D self-adjoint unitaries}
    The space $\calSU^{\calL_\sharp}_\LamNT$ is path-connected for $\sharp\in\Set{\mathrm{I,II}}$.
    \end{thm}
    \begin{proof}
    Throughout the proof, we fix an ONB $\{\vf^j_k\}_{k=1}^\infty$ for each $\im\Lambda_j$ and let $E^j_k=\vf^j_k\otimes (\vf^j_k)^\ast$.
    
    We first prove that $\calSU^{{\calL_{\mathrm{I}}}}_\LamNT$ is path-connected. The proof consists of five steps:
    \begin{itemize}
        \item \emph{Step 1}: Let $U\in \calSU^{\calL_\mathrm{I}}_\LamNT$. Define the canonical operator $X$ to be \eql{\label{eq:canonical operator}
            X\vf^j_k = (-1)^{k}\vf^j_k,\qquad (j\in\calJ,k\in\NN)\,.
        }
        Then $X\in\calSU^{{\calL_{\mathrm{I}}}}_\LamNT$. We will deform $U\leadsto X+B$ within $\calSG^{{\calL_{\mathrm{I}}}}$ where $B\in\calK_I$.
        \item \emph{Step 2}: There exists a projection $P$ of the form $P=\sum_{j=1}^\infty \sum_{k=1}^{n_j} E_k^j$ for some $n_j\in\NN$. We will deform $X+B\leadsto P^\perp X P^\perp + P(X+B)P$ within $\calSG^{{\calL_{\mathrm{I}}}}$.
        \item \emph{Step 3}: Construct $PYP:\im P\to \im P$ that commutes with $E_k^j$ for $j\in\calJ$ and $1\leq k\leq n_j$ and satisfies \eql{\label{eq:dimensions of PYP and P(X+B)P}
            \dim\ker(\sgn(PYP)\pm \Id_{\im P}) = \dim\ker(\sgn(P(X+B)P\pm \Id_{\im P}:\im P\to\im P))\,.
        } We will deform $P(X+B)P\leadsto Y$ within $\calSG(\im P)$.
        \item \emph{Step 4}: Deform $P^\perp XP^\perp+PYP\leadsto X$ within $\calSG^{{\calL_{\mathrm{I}}}}$. 
        \item \emph{Step 5}: Lift the paths from $\calSG^{{\calL_{\mathrm{I}}}}$ to $\calSU^{{\calL_{\mathrm{I}}}}$.
    \end{itemize}

    \paragraph{Step 1}
    Let $U\in \calSU^{\calL_\mathrm{I}}_\LamNT$. The operator $U_j$ is Fredholm by \cref{rem:local SAI Fredholm} and it has $\findex U_j=0$ by sef-adjointness. In particular, we can write \eql{\label{eq:U_j is Fredholm zero}
        U_j=G_j+K_j \qquad(j\in\calJ)
    } for some non-trivial SAI $G_j$ and some $K_j\in\calK$ (indeed, $G_j$ may be chosen not only invertible but also self-adjoint). The non-triviality of $G_j$ comes from the fact that $U$ is $\Lambda$-non-trivial. Indeed \eq{
        \sigmaess(\sgn(G_j)) \equiv \sigma(\pi(\sgn(G_j)))=\sigma(\sgn(\pi(G_j)))=\sigma(\sgn(\pi(U_j))) = \Set{\pm 1}\,.
    }
    See also \cref{rem:non-triviality when A_j is invertible}. Let $G:=\bigoplus_j G_j$ and also define $A:=U-G$. We remark that $A$ is not necessarily compact, but it is in fact in $\calK_I$ (recall that space from \cref{prop:quasi compcat}). Indeed, we have $A=U-G\in{\calL_\mathrm{I}}$. Thus $\Lambda_j^\perp A\Lambda_j$ and $\Lambda_jA\Lambda_j^\perp$ are compact. It follows from \eq{
        \Lambda_jA\Lambda_j=K_j\in\calK\tag{Using \cref{eq:U_j is Fredholm zero}}
    } that $\Lambda_j A=\Lambda_jA\Lambda_j^\perp + \Lambda_jA\Lambda_j \in\calK$
    and similarly for $A\Lambda_j\in\calK$. Let $X$ be a diagonal operator \cref{eq:canonical operator}. It is clear by construction that $X$ is a $\Lambda$-non-trivial SAU. 
    Since $G_j$ and $X_j$ are both non-trivial SAIs, using \cref{lem:path between SAIs} right below, there exists SAIs $F_j$ such that $G_j=F_j^*X_jF_j$. Let $F$ be the direct sum of $F_j$. Let $B=(F^\ast)^{-1}AF^{-1}$. Then
    \eq{
        U=G+A = F^\ast (X+B) F \leadsto X+B
    }
    where we deform each $F_j\leadsto \Id$ within invertible operators. Since $[F,\Lambda_j]=0$, it follows that the path lies in $\calSG^{\calL_\mathrm{I}}$. We have $\Lambda_jB,B\Lambda_j\in\calK$ since $B\in{\calL_\mathrm{I}}$ and $\Lambda_jB\Lambda_j=(F^\ast)^{-1}\Lambda_jA\Lambda_jF^{-1}\in\calK$, i.e., $B\in\calK_I$ and hence we deformed $U$ into the space $X+\calK_I$.

    \paragraph{Step 2}
    Using \cref{lem:compact operators}, there exists a projection $P$ of the form $P=\sum_{j=1}^\infty \sum_{k=1}^{n_j} E_k^j$ for some $n_j\in\NN$ such that $\norm{B-PBP}<\ve$ and \eql{\label{eq:PTP}
        PTP \in {\calL_{\mathrm{I}}} \qquad(T\in\calB)\,.
    } 
    Using the fact that the space of invertible operators is open, for $\ve$ sufficiently small, we can use the straight-line homotopy to deform $X+B$ to $X+PBP$ within $\calSG^{\calL_\mathrm{I}}$. In particular, since $X$ is diagonal, we have \eq{X+PBP=P^\perp X P^\perp + P(X+B)P\,.} 
    Thus $P(X+B)P$ is a SAI on $\im P$.
    
    \paragraph{Step 3}
    Construct $PYP:\im P\to \im P$ that commutes with $E_k^j$ for $j\in\calJ$ and $1\leq k\leq n_j$, and satisfies \cref{eq:dimensions of PYP and P(X+B)P}. Indeed, we can define $PYP \vf_k^j = y_k^j\vf_k^j$ where $y_k^j\in\Set{\pm1}$ such that $\#\Set{y_k^j|y_k^j=\pm 1}$ is equal to the expression in \cref{eq:dimensions of PYP and P(X+B)P}. This condition allows us to use \cref{lem:path between SAIs} to deform $P(X+B)P\leadsto PYP$ within $\calSG(\im P)$.

    \paragraph{Step 4}
    We want to deform $P^\perp XP^\perp Y + PYP$ to $X$. Let $Z_j:=\Lambda_j(P^\perp XP^\perp Y + PYP)\Lambda_j:\im\Lambda_j\to\im\Lambda_j$. Since $\Lambda_jP=P\Lambda_j=\sum_{k=1}^{n_j}E_k^j$ consists of finitely many rank-one projections, and by the non-trivial construction of $X_j$ in \cref{eq:canonical operator}, it follows that $Z_j$ is a non-trivial SAU. Therefore, we can deform $Z_j\leadsto X_j$ within $\calSG(\im \Lambda_j)$ using \cref{lem:path between SAIs}. The direct sum of these paths deforms $P^\perp XP^\perp Y + PYP$ to $X$ within $\calSU^{{\calL_{\mathrm{I}}}}$. The Type-I locality is preserved since the path commutes with $\im\Lambda_j$ for all $j\in\calJ$.

    \paragraph{Step 5}
    Let $\gamma_t:U\leadsto X$ be the path that combines the paths constructed in the previous paragraphs. So far, we have $\gamma_0=U\in\calSU^{{\calL_{\mathrm{I}}}}_\LamNT$ and $\gamma_t\in\calSG^{{\calL_{\mathrm{I}}}}$. We can use \cref{lem:lift SAI local path to SAU local path} to lift the path $\gamma_t$ to one that lies in $\calSU^{{\calL_{\mathrm{I}}}}_\LamNT$. This concludes the proof for the path-connectedness of $\calSU^{{\calL_{\mathrm{I}}}}_\LamNT$.

    We now show that $\calSU^{{\calL_\mathrm{II}}}_\LamNT$ is path-connected. Let $U\in \calSU^{{\calL_\mathrm{II}}}$. We can write $U=\bbDelta U+K$ and $K\in\calK$. Since $K$ is compact, using \cref{lem:compact operators}, for any $\ve>0$ there is $m$ sufficiently large so that with $P:=\sum_{j=1}^m\Lambda_j$ such that $\norm{K-PKP}<\ve$. If we choose $\ve$ small enough, then by the openness of the space of invertible operators, we can deform $U\leadsto \bbDelta U +PAP$ in $\calSG^{{\calL_\mathrm{II}}}$ via the straight-line homotopy (since the sum of two local operators is local). Since $\bbDelta U$ commutes with $\Lambda_j$, we have
    \eq{
        \bbDelta U+PAP = P^\perp \bbDelta U P^\perp + P(\bbDelta U +A)P\,.
    }
    This implies that $P^\perp \bbDelta U P^\perp\in \calSG^{{\calL_\mathrm{II}}}$ and $(\bbDelta U)_j:=\Lambda_j\bbDelta U\Lambda_j:\im\Lambda_j\to\im\Lambda_j$ is a non-trivial SAI for each $j\geq m+1$. Zooming into the subspace $\im P$, we see that $P(\bbDelta U+A)P\in\calSG^{\calL_\mathrm{I}}$ where the locality is with respect to the decomposition $\Id_{\im P}=\sum_{j=1}^m\Lambda_j$. We can thus use the result for $\calSG^{\calL_\mathrm{I}}$ before (the theorem is stated for $\calSU^{\calL_\mathrm{I}}$ but the proof works just as well for $\calSG^{\calL_\mathrm{I}}$), and deform $P(\bbDelta U+A)P$ to $PXP$ for some $X$ that commutes with $\Lambda_j$ and has each $X_j$ being non-trivial SAI. The path stays type-II local since the type-II locality coincides with type-I locality when there are finitely many $\Lambda_j$'s.
    In summary, we have deformed $U$ to a $\Lambda$-non-trivial SAI $Y$ within $\calSG^{{\calL_\mathrm{II}}}$ such that $[Y,\Lambda_j]=0$; such operators are path-connected. We can lift this path in $\calSG^{{\calL_\mathrm{II}}}$ to a path in $\calSU^{{\calL_\mathrm{II}}}$ using \cref{lem:lift SAI local path to SAU local path}.
    
    \end{proof}

    \begin{lem}\label{lem:path between SAIs}
        Let $A$ and $B$ be SAIs. If there is a continuous path $A\leadsto B$ within SAIs then either both $A$ and $B$ are non-trivial, or \eql{\label{eq:same dimensions of pm1 eigenspaces}
        \dim\ker (\sgn(A)\pm \Id)=\dim\ker (\sgn(B)\pm \Id)\,.
        }
        The converse is true and we can write $A=G^\ast B G$ for some invertible operator $G$.
    \end{lem}
    \begin{proof}
    Suppose $A$ and $B$ be SAIs and there is a continuous path $\gamma_t:A\leadsto B$ within SAIs. Then $\sgn(\gamma_t)$ is a path deforming $\sgn(A)\leadsto \sgn(B)$ in SAUs. It follows from \cite[Proposition 2.2.6]{rordam2000introduction} that since $\calB$ is a C-star algebra, there is a unitary $W$ such that $W^\ast \sgn(B) W=\sgn(A)$. Thus $\dim\ker (\sgn(B)\pm \Id) = \dim\ker(\sgn(A)\pm\Id)$. 

    Conversely, suppose $A$ and $B$ are either both non-trivial, or satisfy \cref{eq:same dimensions of pm1 eigenspaces}. Write $A=\sgn(A)|A|$ and $B=\sgn(B)|B|$. In either case, there exists a unitary operator $W$ such that $\sgn(A)=W^\ast \sgn(B) W$. In particular, we have
    \eq{
        A=|A|^{1/2}\sgn(A)|A|^{1/2} = |A|^{1/2}W^\ast \sgn(B) W|A|^{1/2} = |A|^{1/2}W^\ast |B|^{-1/2}B|B|^{-1/2} W|A|^{1/2}\,.
    }
    If we let $G=|B|^{-1/2} W|A|^{1/2}$, then we have $A=G^\ast BG$. Using $G\leadsto \Id$ within the space of invertible operators, we have a continuous path in $\calSG$ deforming $A\leadsto B$.
    \end{proof}
    
    \begin{lem}\label{lem:lambda-non-triviality preserves in continuous path} 
        Let $\sharp\in\Set{\mathrm{I,II}}$. If $A\in\calSG_{\LamNT}^{\calL_\sharp}$ and $B\in\calSG^{\calL_\sharp}$ and there is a path $A\leadsto B$ within $\calSG^{\calL_\sharp}$, then $B\in \calSG_\LamNT^{\calL_\sharp}$.
    \end{lem}
    \begin{proof}
        Let $\sharp\in\Set{\mathrm{I,II}}$, and $A\in\calSG^{\calL_\sharp}_\LamNT$ and $B\in\calSG^{\calL_\sharp}$ and $\gamma_t:A\leadsto B$ in $\calSG^{\calL_\sharp}$. Let $\pi$ be the projection from $\calB$ onto the Calkin C-star algebra. Let $A_j,B_j,(\gamma_j)_t$ be operators $\Lambda_j A\Lambda_j,\Lambda_jB\Lambda_j,\Lambda_j\gamma_t\Lambda_j$ viewed as operators $\im\Lambda_j\to\im\Lambda_j$. Since $A,B,\gamma_t\in\calSG^{\calL_\sharp}$, it follows that $A_j,B_j$ and $(\gamma_j)_t$ are Fredholm. Thus, $\pi(A_j),\pi(B_j),\pi((\gamma_j)_t)$ are SAIs with $\pi((\gamma_j)_t)$ deforming $\pi(A_j)\leadsto \pi(B_j)$. In particular, $\sgn(\pi((\gamma_j)_t))$ is a continuous path (of SAUs in the Calkin algebra) that deforms $\sgn(\pi(A_j))\leadsto \sgn(\pi(B_j))$. We use again \cite[Proposition 2.2.6]{rordam2000introduction} on the Calkin C-star algebra to get a unitary element $u$ in the Calkin algebra such that $u^*\sgn(\pi(B_j)) u=\sgn(\pi(A_j))$. Since $\pm 1$ lies in the spectrum of $\sgn(\pi(A_j))$, it follows that $\pm 1$ is present in the spectrum of $\sgn(\pi(B_j))$.

    \end{proof}

    \begin{lem}\label{lem:lift SAI local path to SAU local path}
    Let $\sharp\in\Set{\mathrm{I,II}}$. Suppose $U,V\in \calSU^{\calL_\sharp}_\LamNT$ and we have a path $U\leadsto V$ within $\calSG^{\calL_\sharp}$. Then there is a path $U\leadsto V$ within $\calSU^{\calL_\sharp}_\LamNT$.
    \end{lem}
    \begin{proof}
    Let $\gamma_t$ be the path $U\leadsto V$ in $\calSG^{\calL_\sharp}$. Since $\sgn$ is continuous on the spectrum of $\gamma_t$ which has a gap, it follows that $\sgn(\gamma_t)$ is continuous and deform $U\leadsto V$ within $\calSU^{\calL_\sharp}$. Moreover, using \cref{lem:lambda-non-triviality preserves in continuous path}, the path lies in $\calSU^{\calL_\sharp}_\LamNT$.
    \end{proof}

    \section{Physical consequences}\label{sec:physical consequences}
    Let us briefly indicate what are the physical implications of the foregoing sections.

    \subsection{Star-shaped-wire chiral systems}
    Consider a physical system in the shape of a star with $k$-legs, $k\in\NN_{\geq3}$ (the case $k=2$ is a one-dimensional system, already dealt with in \cite{ChungShapiro2023}). Its Hilbert space is \eq{
        \calH_k:=\ell^2(\NN\times\Set{1,\dots,k})\otimes \CC^N\,.
    } While this Hilbert space is isomorphic to \eq{
        \ell^2(\NN)\otimes\CC^k\otimes\CC^N
    } we want to emphasize a particular mode of locality and for that reason we prefer not to use this isomorphism. Indeed, in writing the latter we imply that electrons may hop, on one one-dimensional strip, between any $k\times N$ degrees of freedom. This is depicted in \cref{fig:incorrect wire}; here, we have $k=3$ and $N=1$, and the electron may hop freely within each red dimer, e.g., the hop indicated by the blue arrow. On the other hand, in the former we mean that there would be spatial distance between the various $k$ degrees of freedom, because they sit on different legs of a system. This is depicted in \cref{fig:correct wire} when $k=3$ and $N=1$. Using this picture, a very physical form of locality in this system would be the graph distance between any two points. For instance, we may embed this system within $\ZZ^2$ just so long as each leg sits on a distinct ray.

    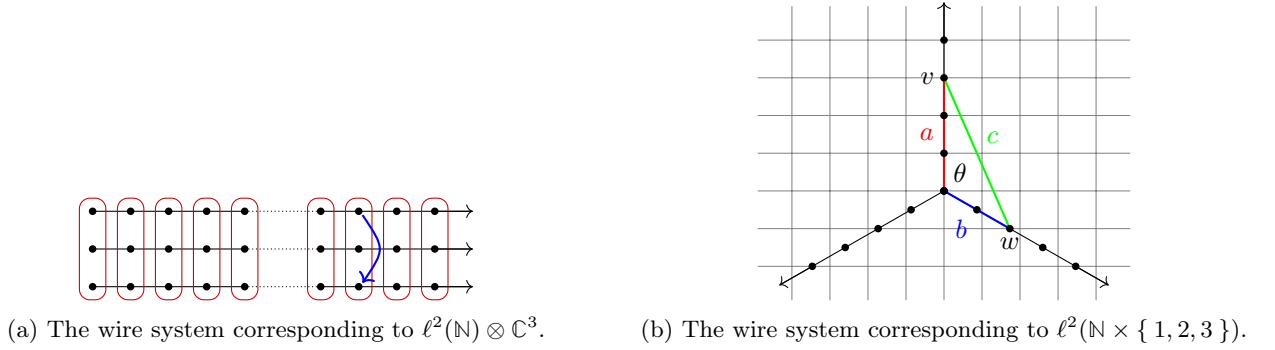
\begin{figure}[t]
    \begin{subfigure}{0.5\textwidth}
    	\centering
        \begin{tikzpicture}[scale=0.5,draw=black]
        \foreach \x in {0,1,2,3,4,6,7,8,9}{
            \foreach \y in {0,1,2}{
                \node (\x\y) at (\x,\y) [circle,fill,inner sep=1pt,black] {};
                \ifthenelse{\NOT\x=4}{\draw[black] (\x,\y)  -- (\x+1,\y)}{};
            }
            \node[draw=red!80!black,rounded corners,fit=(\x0)(\x1)(\x2)] {};
            
        }
        \foreach \y in {0,1,2}{
            \draw[densely dotted] (4,\y) -- (6,\y);
            \draw[->](9,\y) -- (10,\y);
        }
        \draw[blue,thick,->] (7+0.1,2-0.1) .. controls (7.7,1) .. (7+0.1,0+0.1);
        \end{tikzpicture}
            \caption{The wire system corresponding to $\ell^2(\NN)\otimes \CC^3$.}
        \label{fig:incorrect wire}
        \end{subfigure}
    \qquad
    \begin{subfigure}{0.5\textwidth}
        \centering  
        \begin{tikzpicture}[scale=0.5,draw=black]
        \draw[gray,very thin] (-4.9,-2.9) grid (4.9,4.9);
        \foreach \x in {0,1,2,3,4}{
            \draw (0,\x)  -- (0,\x+1);
            \draw ({\x*1.7320508075/2},-\x/2)  -- ({(\x+1)*1.7320508075/2},{-(\x+1)/2});
            \draw ({-\x*1.7320508075/2},-\x/2) -- ({-(\x+1)*1.7320508075/2},{-(\x+1)/2});   
        }
        
        \draw[red,thick] (0,0) -- (0,3) node[midway,left] {$a$};
        \draw[blue,thick] (0,0) -- ({2*1.7320508075/2},-2/2) node[midway,below left] {$b$};
        \draw[green,thick] (0,3) -- ({2*1.7320508075/2},-2/2) node[midway,above right] {$c$};
        \draw (0,0) node[black,above right]{$\theta$};
        \node at (0,3) [left,black]{$v$};
        \draw ({2*1.7320508075/2},-2/2) node[below,black]{$w$};

        \foreach \x in {0,1,2,3,4}{
            \draw (0,\x)  node[circle,fill,inner sep=1pt,black] {};
            \draw ({\x*1.7320508075/2},-\x/2) node[circle,fill,inner sep=1pt,black] {} ;
            \draw ({-\x*1.7320508075/2},-\x/2) node[circle,fill,inner sep=1pt,black] {};   
        }

        \def\x{4}
        \draw[->] (0,\x)  -- (0,\x+1);
        \draw[->] ({\x*1.7320508075/2},-\x/2)  -- ({(\x+1)*1.7320508075/2},{-(\x+1)/2});
        \draw[->] ({-\x*1.7320508075/2},-\x/2) -- ({-(\x+1)*1.7320508075/2},{-(\x+1)/2});   

        \end{tikzpicture} 
          
        \caption{The wire system corresponding to $\ell^2(\NN\times\Set{1,2,3})$.}
        \label{fig:correct wire}
        \end{subfigure}
    
    \caption{Two versions of wire systems.}
    \label{fig:wires}
    \end{figure}

    \begin{lem}
        The graph distance is comparable to the Euclidean distance if the graph is embedded in $\ZZ^2$.
    \end{lem}
    \begin{proof}
        For two vertices on the same leg, their Euclidean distance and graph distance are the same. Consider two vertices $v$ and $w$ on different legs (see \cref{fig:correct wire}.) The graph distance between $v$ and $w$ is $a+b$, while their Euclidean distance is $c$. There exists an angle-dependent constant $D(\theta)<\infty$ such that
        \eq{
            c\leq a + b \leq D(\theta) c\,.
        }
        Using the law of sines and cosines, we can roughly estimate the constant to be $D(\theta)=\sqrt{1+2/|\sin\theta|}$ which remains finite if $\theta>0$. Therefore, the two notions of distance are comparable.
    \end{proof}

    On $\calH_k\equiv\ell^2(\NN\times\Set{1,\dots,k})\otimes \CC^N$, we say that $A$ is exponentially local if there exists some $C<\infty$, $\mu>0$ such that
    \eq{
       \norm{A_{xy}}\leq C\exp(-\mu d(x,y)) \qquad (x,y\in \NN\times \Set{1,\dots,k})\,.
    }
    Here the distance $d(x,y)$ can be the graph distance or the Euclidean distance if we embed $\NN\times \Set{1,\dots,k}$ within $\ZZ^2$.

    \begin{lem}
        Any one of these two notions of locality (they are the same as they are comparable) implies $\calL_\mathrm{I}$ or $\calL_\mathrm{II}$ (they are the same since $k<\infty$).
    \end{lem}
    \begin{proof}
        Suppose $A$ is exponentially local and we use the graph distance. Let $j\in\Set{1,\dots,k}$. We will show that $[\Lambda_j,A]$ is trace-class. This is standard, see e.g. \cite[Lemma 2.(b)]{Graf_Shapiro_2018_1D_Chiral_BEC}. We repeat it for convenience:
        \eq{
            \sum_{x,y\in\NN\times\Set{1,\dots,k}} \norm{[\Lambda_j,A]_{xy}}&=\sum_{x,y\in\NN\times\Set{1,\dots,k}} |\Lambda_j(x)-\Lambda_j(y)|\norm{A_{xy}} \\
            &\leq 2C\sum_{x\in\NN\times\Set{j}}\sum_{y\in\NN\times\Set{1,\dots,j-1,j+1,\dots,k}} \exp(-\mu d(x,y)) \\
            &=2Ck\sum_{x\in\NN} e^{-\mu x} \sum_{y\in\NN}e^{-\mu y} < \infty\,.
        }

        Since $j\in\Set{1,\dots,k}$ was arbitrary, this implies ${\calL_{\mathrm{I}}}$. Since $k<\infty$, ${\calL_{\mathrm{I}}}={\calL_{\mathrm{II}}}$ and we are finished.
    \end{proof}

    \begin{cor}Chiral systems, i.e., systems of the form \eq{
        H = \begin{bmatrix}
            0 & S^\ast\\
            S & 0
        \end{bmatrix}\,,
    } where $S$ is some not-necessarily self-adjoint operator on $\calH_k$ which is invertible, has $k-1$ indices, each of which is $\ZZ$-valued. Moreover, we have proven above that the set of path-connected components of such systems (if we are allowed to make deformations within $\calL_\sharp$, $\sharp=\mathrm{I},\mathrm{II}$) is $\ZZ^{k-1}$. It is clear that this index has something to do with the chiral number of edge states at the central vertex. We are not aware that such an experiment with wires of Polyacetylene has been carried out.
    \end{cor}
    \subsection{Two-dimensional systems}
    We now turn to actual two-dimensional systems on $\ZZ^2$. Here of course the initial notion of locality to employ is exponential locality inherited from $\ZZ^2$. As has been mentioned, this form of locality implies $L$-locality (where $L$ is the Laughlin flux insertion operator \cref{eq:Laughlin flux operator}).
    \begin{example}
        Exponential locality does not imply ${\calL_{\mathrm{I}}}$ locality.
    \end{example}
    \begin{figure}[t]
        \centering
        \begin{tikzpicture}[scale=0.5]
        \draw[step=1cm,gray,opacity=0.5,very thin] (-3.9,-1.9) grid (3.9,5.9);
        \draw[thick,->,black] (0,0) -- (4.2,0);
        \draw[thick,->,black] (0,0) -- (0,6.2);
        
        \foreach \x in {-4,...,3}{
        \foreach \y in {-1,...,5}{
        \draw[very thick,->,teal] (\x+0.2,\y) -- (\x+1-0.2,\y);
        }
        }
        \end{tikzpicture}
    \caption{Bilateral right shift on $\ZZ^2$.}
    \label{fig:2d right shift}
    \end{figure}
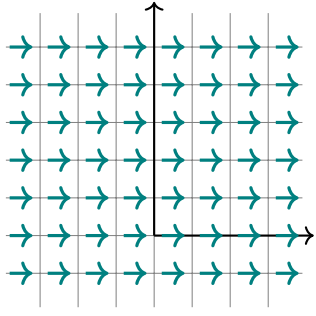
    \begin{proof}
        Consider the right shift operator $R_1$ in \cref{ex:counterexample}, see \cref{fig:2d right shift}. Let $\calJ=\bS^1_\QQ$ and consider $Q:=\Lambda_{\{\pi/2\}}\equiv\chi_{\{\pi/2\}}(L)$. Since $R_1$ is finite-hopping, it is exponentially local. If $R_1\in\calL_\mathrm{I}$, then $[R_1,Q]\in\calK$ and $Q^\perp[R_1,Q]=Q^\perp R_1Q\in\calK$. However, the operator $Q^\perp R_1 Q$ shifts infinitely many position basis elements on the $y$-axis to the right, see \cref{fig:2d right shift}, and hence cannot be compact. More formally, we have that $R_1^\ast Q^\perp R_1 Q = Q\in\calK$. However, $Q$ cannot be compact since $\im Q$ is infinite-dimensional. 
    \end{proof}
    Thus, we have presented a counterexample that is exponentially local but is not in $\calL_\mathrm{I}$. Since $\calL_{\mathrm{II}}\subseteq\calL_{\mathrm{III}}\subseteq \calL_\mathrm{I}$ (see \cref{fig:flow chart localities}), it follows that exponential locality also does not imply $\calL_\mathrm{II}$ nor $\calL_{\mathrm{III}}$.

    As such, it is only the somewhat unphysical, \emph{quasi-2D-locality-obeying} systems that have an infinite number of indices. It is not clear to us how to implement such systems physically or what kind of Hamiltonians they correspond to.

	\bigskip
	\bigskip
	\noindent\textbf{Acknowledgments.} 
	We are  indebted to Gian Michele Graf for stimulating discussions. J.H.C. thanks Shouda Wang for helpful discussions.
	\bigskip
	\appendix

	\section{Operator Integrals}\label{sec:operator integrals}

    Let $(\Omega,\nu)$ be some complex measure space. Let $F:\Omega\to \calB(\calH)$ be an operator-valued function. Let us recall how the (weak) integral $\int_\Omega F(w)d\nu(w)\in\calB(\calH)$ can be constructed. We refer the reader to \cite[pp. 77]{rudin1991functional} for a more general discussion (there, consider the space $\calB(\calH)$ endowed with strong operator topology). We say that $F$ is weakly measurable if $w\mapsto \braket{\psi,F(w)\vf}$ is measurable for all $\psi,\vf\in\calH$. We say that $F$ is integrable if $\int_\Omega \norm{F(w)}\dif |\nu|(w)<\infty$ where $\norm{F(w)}$ is the operator norm. It is clear that $(\psi,\vf)\mapsto \int \braket{\psi,F(w)\vf}\dif \nu (w)$ is a bounded sesquilinear form. Thus, by Riesz representation theorem, there exists a bounded operator $L$ such that $\braket{\psi,L\vf} = \int \braket{\psi,F(w)\vf}\dif \nu (w)$ for all $\psi,\vf\in\calH$. We write $L=\int_\Omega F(w)\dif \nu(w)$ and call it the integral of $F$.

    \begin{prop}\label{prop:properties of weak integrals}
    Here are some properties for the operator integrals. \begin{enumerate}
    \item $\norm{\int F \dif\nu } \leq \int \norm{F}\dif|\nu|$.
    \item $\norm{\int F \dif\nu }_1 \leq \int \norm{F}_1\dif|\nu|$ where $\norm{\cdot}_1$ is the trace-class norm.
    \item 
    If $F$ maps to the trace-class operators and $\int F\dif \nu$ is also trace class and $\int \norm{F}_1\dif |\nu|<\infty$, then
        \eq{
            \tr \int F\dif \nu = \int \tr F \dif \nu\,.
        }
    \item If $T$ is any bounded operator, then $T\int F(w) \dif \nu(w) = \int TF(w) \dif \nu(w)$.
    \end{enumerate}
    \end{prop}
    \begin{proof}
    For item 2., recall that $\norm{T}_1 = \sup_X |\tr (XT)|/\norm{X}$ where the supremum is taken over all finite-rank operators $X$ (see, e.g., \cite[Lemma 10, Chapter 11]{birman2012spectral}.) Finite-rank operators are of the form $X=\sum_{i=1}^n\alpha_i\vf_i\otimes \psi_i^\ast$. We have
    \eq{
        \left|\tr X\int F\dif\nu \right| = \left|\sum_{i=1}^n\alpha_i \tr \vf_i\otimes \psi_i^\ast \int F\dif\nu \right|  = \left|\sum_{i=1}^n\alpha_i \int \braket{\psi_i,F\vf_i}\dif \nu\right| = \left|\int \tr X F \dif \nu\right| \leq \int |\tr XF|\dif|\nu|\,.
    }
    We can take supremum over all finite-rank operators to conclude the result.
    
    Let us prove item 3. Let $\{\vf_i\}_{i=1}^\infty$ be an orthonormal basis. Then 
    \eq{
        \tr \int F\dif \nu = \sum_{i=1}^\infty \int \braket{\vf_i,F\vf_i} \dif \nu\,.
    }
    We have $\left|\sum_{i=1}^n \braket{\vf_i,F(w)\vf_i}\right| \leq \norm{F(w)}_1$ and $\int \norm{F}_1 \dif |\nu|<\infty$. Thus, by the dominated convergence theorem, we can interchange the sum and integral above. In particular, we have $\sum_{i=1}^\infty \braket{\vf_i,F\vf_i}=\tr F$.
    \end{proof}

    \begin{prop}\label{prop:double operator integrals over compact operators}
    Let $F,G:\Omega\to\calB(\calH)$ be weakly measurable such that
    \eql{
        C=\int_\Omega \norm{F(w)}\norm{G(w)} \dif{|\nu|(w)} < \infty \label{eq:integralbility condition}
    }
    and let $A\in\calB(\calH)$ be arbitrary. Then \begin{enumerate}
        \item $w\mapsto F(w)AG(w)$ is integrable. 
        \item if $A$ is trace-class, then so is $\int_\Omega F(w) AG(w)\dif\nu(w)$ and we have $\norm{\int_\Omega F(w) AG(w)\dif\nu(w)}_1\leq C\norm{A}_1$. 
        \item if $A$ is compact, then so is $\int_\Omega F(w) AG(w)\dif\nu(w)$.
    \end{enumerate} 
    \end{prop}
    \begin{proof}
    The map $w\mapsto F(w)AG(w)$ is integrable. Indeed, we have $\int_\Omega \norm{F(w)AG(w)} \dif{|\nu|}(w) \leq C\norm{A}$. If $A$ is trace-class, using \cref{prop:properties of weak integrals}, we have
    \eq{
        \norm{\int F(w)AG(w)\dif \nu(w)}_1\leq \int \norm{F(w)AG(w)}_1 \dif|\nu|(w) \leq \int \norm{F(w)}\norm{G(w)}\norm{A}_1\dif |\nu|(w) \leq C\norm{T}_1
    }
    where we have used the fact that $\norm{XYZ}\leq \norm{X}\norm{Z}\norm{Y}_1$ if $X,Z\in\calB$ and $Y\in\calT_1$ (see, e.g., \cite[Eq. (2), Section 2, Chapter 11]{birman2012spectral}.)

    If $A$ is compact, then there is a sequence of finite-rank operators $A_n$ that converges to $A$ in operator norm. Thus $\int F(w)A_nG(w)\dif\nu(w)$ converges to $\int F(w)AG(w)\dif\nu(w)$ in operator norm with the estimate \eq{
        \norm{\int F(w)A_nG(w)\dif\nu(w)-\int F(w)AG(w)\dif\nu(w)} &\leq \int \norm{F(w)(A_n-A)G(w)} \dif|\nu|(w) \\
        &\leq C\norm{A_n-A}\,.
    } Since $\int F(w)A_nG(w)\dif\nu(w)$ is in fact trace-class, it follows that $\int F(w)AG(w)\dif\nu(w)$ is compact.
    \end{proof}

		\begingroup
		\let\itshape\upshape
		\printbibliography
		\endgroup
	\end{document}